%% file: public-chance-correlation.tex
\documentclass[letterpaper]{article}

\usepackage[preprint,nonatbib]{neurips_2020}
\usepackage[utf8]{inputenc} 
\usepackage[T1]{fontenc}    
\usepackage{hyperref}       
\usepackage{url}            
\usepackage{booktabs}       
\usepackage{amsfonts}       
\usepackage{nicefrac}       
\usepackage{microtype}      
\usepackage{times}
\usepackage{amsmath}
\usepackage{amssymb}
\usepackage{amsthm}
\usepackage{microtype}
\usepackage{enumitem}
\usepackage{dsfont}
\usepackage[capitalise,noabbrev]{cleveref}
\usepackage{xcolor}
\usepackage{etoolbox}
\usepackage{thm-restate}
\usepackage[linesnumbered,ruled,lined,noend]{algorithm2e}
\usepackage{amsmath}
\usepackage{mathtools}
\usepackage{amsthm}
\usepackage{mleftright}
\usepackage{stmaryrd}
\usepackage{float}
\usepackage{bm}
\usepackage{numprint}
\usepackage{tikz}
\usepackage{cleveref}
\usepackage{standalone}
\usepackage{multirow}
\usepackage[numbers]{natbib}
\usepackage{siunitx}

\usepgflibrary{shapes.geometric,arrows.meta}

\DeclareMathOperator*{\ext}{\,\triangleleft\,}

\ifcsmacro{assumption}{}{
	
}
\ifcsmacro{corollary}{}{
	
}
\ifcsmacro{definition}{}{
	\newtheorem{definition}{Definition}[]
}
\ifcsmacro{example}{}{
	
}
\ifcsmacro{fact}{}{
	
}
\ifcsmacro{informal_theorem}{}{
	
}
\ifcsmacro{lemma}{}{
	\newtheorem{lemma}{Lemma}[]
}
\ifcsmacro{proposition}{}{
	
}
\ifcsmacro{remark}{}{
	\newtheorem{remark}{Remark}[]
}
\ifcsmacro{theorem}{}{
	
}
\ifcsmacro{observation}{}{
	
}
\mathtoolsset{centercolon}

\makeatletter
\patchcmd\algocf@Vline{\vrule}{\vrule \kern-0.4pt}{}{}
\patchcmd\algocf@Vsline{\vrule}{\vrule \kern-0.4pt}{}{}
\makeatother

\SetKwComment{Hline}{}{\vspace{-3mm}\textcolor{gray}{\hrule}\vspace{1mm}}
\definecolor{darkgrey}{gray}{0.3}
\definecolor{commentcolor}{gray}{0.5}
\SetKwComment{Comment}{\color{commentcolor}$\triangleright$\ }{}
\SetCommentSty{}
\SetNlSty{}{\color{darkgrey}}{}
\crefname{algocf}{Algorithm}{Algorithms}
\SetKwProg{Fn}{function}{}{}
\SetKwProg{Subr}{subroutine}{}{}

\newcommand{\defeq}{\mathrel{:\mkern-0.25mu=}}

\newcommand{\cX}{\mathcal{X}}
\newcommand{\cY}{\mathcal{Y}}

\newcommand{\cJ}{\mathcal{J}}
\newcommand{\cK}{\mathcal{K}}

\newcommand{\cS}{\mathcal{S}}
\newcommand{\cD}{\mathcal{D}}

\newcommand{\bbR}{\mathbb{R}}

\renewcommand{\vec}[1]{\bm{#1}}

\DeclareMathOperator{\img}{Im}

\DeclareMathOperator{\co}{co}

\DeclareMathOperator{\dep}{depth}

\newcommand{\infos}[1]{\mathcal{I}_{#1}}
\newcommand{\emptyseq}{\varnothing}

\usepackage{pifont}

\newcommand\circled[1]{\tikz[baseline=(char.base)]{
            \node[shape=circle,draw,inner sep=.5pt] (char) {\fontsize{6.5pt}{7pt}\selectfont\normalfont{#1}};}}

\newcommand{\symp}[1]{\Delta^{\!#1}}

\newcommand{\cV}{\mathcal{V}}

\newcommand{\conn}{\rightleftharpoons}
\DeclareMathOperator{\rele}{\bowtie}
\newcommand{\bbone}{\mathds{1}}

\crefalias{AlgoLine}{line}%

\makeatletter
\let\cref@old@stepcounter\stepcounter
\def\stepcounter#1{%
  \cref@old@stepcounter{#1}%
  \cref@constructprefix{#1}{\cref@result}%
  \@ifundefined{cref@#1@alias}%
    {\def\@tempa{#1}}%
    {\def\@tempa{\csname cref@#1@alias\endcsname}}%
  \protected@edef\cref@currentlabel{%
    [\@tempa][\arabic{#1}][\cref@result]%
    \csname p@#1\endcsname\csname the#1\endcsname}}
\makeatother

\title{Polynomial-Time Computation of Optimal Correlated Equilibria in Two-Player Extensive-Form Games with Public Chance Moves and Beyond}
\author{Gabriele Farina\\
Computer Science Department\\
Carnegie Mellon University\\
\texttt{gfarina@cs.cmu.edu}
\And Tuomas Sandholm\\
Computer Science Department, CMU\\
Strategy Robot, Inc.\\
Strategic Machine, Inc.\\
Optimized Markets, Inc.\\
\texttt{sandholm@cs.cmu.edu}
}

\DeclareMathSizes{10}{9}{7}{5}
\makeatletter
\g@addto@macro \normalsize {%
 \addtolength\abovedisplayskip{-3pt}%
 \addtolength\belowdisplayskip{-3pt}%
}
\makeatother

\usetikzlibrary{arrows.meta,calc,shapes.misc,arrows}
\tikzset{cross/.style={path picture={
  \draw[black]
(path picture bounding box.south east) -- (path picture bounding box.north west) (path picture bounding box.south west) -- (path picture bounding box.north east);
}}}
\tikzstyle{chanode}=[fill=white,draw=black,circle,cross,inner sep=.8mm]
\tikzstyle{pl1node}=[fill=black,draw=black,circle,inner sep=.5mm]
\tikzstyle{pl2node}=[fill=white,draw=black,circle,inner sep=.55mm]
\tikzstyle{termina}=[fill=white,draw=black,inner sep=.6mm]

\newcommand{\vsfp}{von Stengel-Forges polytope}
\newcommand{\vsfc}{von Stengel-Forges constraint}

\begin{document}
    \maketitle

    \begin{abstract}
        Unlike normal-form games, where correlated equilibria have been studied for more than 45 years, \emph{extensive-form} correlation is still generally not well understood. Part of the reason for this gap is that the sequential nature of extensive-form games allows for a richness of behaviors and incentives that are not possible in normal-form settings. This richness translates to a significantly different complexity landscape surrounding extensive-form correlated equilibria. As of today, it is known that finding an optimal \emph{extensive-form correlated equilibrium (EFCE)}, \emph{extensive-form coarse correlated equilibrium (EFCCE)}, or \emph{normal-form coarse correlated equilibrium (NFCCE)} in a two-player extensive-form game is computationally tractable when the game does not include chance moves, and intractable when the game involves chance moves. In this paper we significantly refine this complexity threshold by showing that, in two-player games, an optimal correlated equilibrium can be computed in polynomial time, provided that a certain condition is satisfied. We show that the condition holds, for example, when all chance moves are \emph{public}, that is, both players observe all chance moves. This implies that an optimal EFCE, EFCCE and NFCCE can be computed in polynomial time in the game size in two-player games with public chance moves, providing the biggest \emph{positive} complexity result surrounding extensive-form correlation in more than a decade.
    \end{abstract}

    \input{text/introduction}

    \input{text/preliminaries}
    \input{text/decomposition}
    \input{text/integrality_of_V}
    \input{text/experiments}
    \input{text/conclusions}

    \section*{Broader Impact}

     In this paper we give a \emph{positive} complexity result, showing that optimal equilibrium according to three important extensive-form imperfect-information game correlated solution concepts can be computed efficiently in settings---two-player games with public chance moves---where it was generally believed to be impossible. In fact, we showed that this can be done more broadly: in all games where a certain triangle-freeness condition holds.

    Correlated solution concepts have many advantages. First, they enable incentive-compatible coordination of agents. Such coordination is achieved via incentives, rather than forcing: mediators in correlated solution concepts are only able to recommend behavior, but not force it. So, it is up to the mediator to come up with a correlated distribution of recommendations such that no agent has incentive to deviate from the recommendations. Second, in some general-sum interactions these solution concepts are known to enable significantly higher social welfare than Nash equilibrium, while at the same time sidestepping some of the other shortcomings of Nash equilibrium (for example, some equilibrium selection issues).

    In this paper, we are particularly interested in \emph{optimal} correlated equilibria. In other words, our technology can empower the system designer (mediator) to select, among the infinite number of correlated equilibria of the game, one that maximizes a given objective. For example, this technology could be used to find correlated equilibria than maximize social welfare, leading to highest societal good. However, like most technology, our technology has potential for abuse. If used maliciously, the ability to select particular correlated equilibria could be used to \emph{minimize} social welfare, maximize only one of the agent's utility, or minimize all others' utilities---thereby furthering existing inequality or creating new inequality.

    \bibliographystyle{plainnat}
    \bibliography{dairefs}

\iftrue
    \clearpage
    \onecolumn
    \appendix

    \input{text/appendix_proofs}

    \input{text/appendix_experiments}
\fi
\end{document}

%% file: text/introduction.tex
\section{Introduction}
A vast body of literature in computational game theory has focused on computing Nash equilibria (NEs) in two-player zero-sum imperfect-information extensive-form games. Success stories from that endeavor include the creation of strong---in some cases superhuman---AIs for several complex games, including two-player limit Texas hold'em~\cite{Bowling15:Heads}, two-player no-limit Texas hold'em~\cite{Brown17:Safe,Brown17:Superhuman,Moravvcik17:DeepStack}, and multiplayer no-limit Texas hold'em~\cite{Brown19:Superhuman}.
NE captures strategic interactions in which each player maximizes her own utility. The interaction in NE is assumed to be fully decentralized: no communication between players is possible and the behavior of the players is not coordinated by any external orchestrator in any way.
While that assumption is natural in games such as poker, NE is too restrictive in other types of strategic interactions in which partial forms of communication or centralized control are possible~\citep{Ashlagi08:Value}. Therefore, there has been growing interest around less restrictive solution concepts than NE.

\emph{Correlated} and \emph{coarse correlated equilibria} are classic families of solution concepts that relax the assumptions of NE to allow forms of coordination of utility-maximizing agents~\citep{Aumann74:Subjectivity,Moulin78:Strategically}. In correlated and coarse correlated equilibria, a mediator that can recommend behavior---but not enforce it---complements the game. Before the interaction starts, the mediator samples a profile of recommended strategies (one for each player) from a publicly known correlated distribution. The mediator reveals the next recommended move (or sequence of moves, depending on the specific solution concept in the family) to each acting player. In correlated equilibrium, each agent must decide whether to commit to following the next recommended move (or sequence of moves) \emph{after} such move or sequence of moves is revealed by the mediator. In coarse correlated equilibrium, each agent must decide whether to commit to following the next recommended move (or sequence of moves) \emph{before} it is revealed by the mediator.
If a player chooses not to follow the recommendation, the mediator stops issuing further recommendations to that player.
Since the selfish agents are free to not follow the recommendations, it is up to the mediator to come up with a correlated distribution of recommendations such that no agent has incentive to deviate from the recommendations, assuming no other player deviates. Despite the apparent weakness of a mediator that cannot enforce behavior but only suggest it, the maximum social welfare (that is, sum of the players' utilities) that can be induced by these families of solution concepts is greater than the social welfare obtainable by NE.
Examples of interactions where a mediator is natural include traffic control and load balancing~\cite{Ashlagi08:Value}.

These equilibrium concepts have typically been studied in normal-form (that is, matrix) games. The study of correlation in \emph{extensive-form} (that is, tree-form) games is recent, and was pioneered by \citet{Stengel08:Extensive}. Three correlated solution concepts are often used in extensive-form games: \emph{extensive-form correlated equilibrium (EFCE)}~\citep{Stengel08:Extensive}, \emph{extensive-form coarse correlated equilibrium (EFCCE)}~\citep{Farina20:Coarse}, and \emph{normal-form coarse correlated equilibrium (NFCCE)}~\citep{Moulin78:Strategically,Celli19:Computing,Celli19:Learning}.
Compared to normal-form (that is, one-shot) games, extensive-form correlation poses new and different challenges, especially in settings where the agents retain private information. This is unique to the sequential nature of extensive-form games, where, fundamentally,
players can adjust strategically as they make observations about their opponents and the environment~\cite{Farina19:Correlation}.
These challenges also translate to some negative complexity results for extensive-form correlation~\cite{Gilboa89:Nash,Stengel08:Extensive}.
While a landmark positive complexity result in game theory shows that \emph{one} EFCE, EFCCE, or NFCCE can be found in polynomial time~\citep{Papadimitriou05:Computing,Huang08:Computing,Jiang11:Polynomial},
the computation of an \emph{optimal} (that is, one that maximizes or minimizes a given linear objective, such as social welfare) EFCE, EFCCE, or NFCCE is computationally intractable in games with more than two players, as well as two-player games with chance moves, and tractable in two-player games without chance moves~\citep{Stengel08:Extensive}.

In this paper we significantly refine this complexity threshold by showing that, in two-player games, an optimal correlated equilibrium can be computed in polynomial time, provided that a certain \emph{triangle-freeness} condition---which can be checked in polynomial time---is satisfied. We prove that the condition holds, for example, when all chance moves are \emph{public}, that is, both players observe all chance moves. This includes, for example, games where the chance outcomes amount to public dice rolls or public revelations of cards. Specifically, we show that the set of \emph{correlation plans} $\Xi$ of a triangle-free game coincides with the \emph{\vsfp{}} $\cV$ of the game---a polytope that only requires a polynomial number of linear ``probability-mass-conserving'' constraints. Since $\cV$ can be represented using a polynomial number of constraints in the input game size, optimizing over this set can be efficiently done by means of, for example, linear programming methods. 

    \begin{figure}[t]\centering
        \scalebox{.98}{\begin{tikzpicture}
            \tikzstyle{cstep}=[draw] 
              \node[cstep,align=center,text width=2.6cm,minimum height=1cm] at (-.5, 0) {\small Triangle-free game};
              \draw[-implies,double equal sign distance] (-.5, -1.15) -- (-.5,-.65);
              \node at (.4,-.9) {\small\cref{thm:public chance implies triangle freeness}};
              \node[cstep,align=center,text width=2.2cm,minimum height=1cm] at (-.5, -1.8) {\small Game has public chance};
              \draw[-implies,double equal sign distance] (1.15, 0) -- (2.7,0);
              \node[] at (1.9,.3) {\small\cref{thm:decomposition}};
              \node[cstep,align=center,text width=3cm,minimum height=1cm] at (4.5, -1.8) {\small Efficient regret minimizer for $\cV$};
              \draw[implies-,double equal sign distance] (4.5, -1.15) -- (4.5,-.65);
              \node[rotate=0] at (4.95, -.9) {\small\citep{Farina19:Efficient}};
              \node[cstep,align=center,text width=3cm,minimum height=1cm] at (4.5, 0) {\small Scaled-extension-based decomposition of $\cV$};
              \draw[-implies,double equal sign distance] (6.35, 0) -- (7.9,0);
              \node[] at (7.1,.3) {\small\cref{thm:integrality}};
              \node[cstep,align=center,text width=3cm,minimum height=1cm] at (9.7, 0) {\small Integrality of the vertices of $\cV$};
              \draw[implies-implies,double equal sign distance] (9.7, -1.15) -- (9.7,-.65);
              \node at (10.6,-.9) {\small\cref{thm:xi equal vonS}};
              \node[cstep,align=center,text width=2cm,minimum height=1cm] at (9.7 , -1.8) {$\Xi = \cV$};
        \end{tikzpicture}}\vspace{-1mm}
        \caption{Overview of the connections among this paper's results.}
        \label{fig:overview}
    \end{figure}
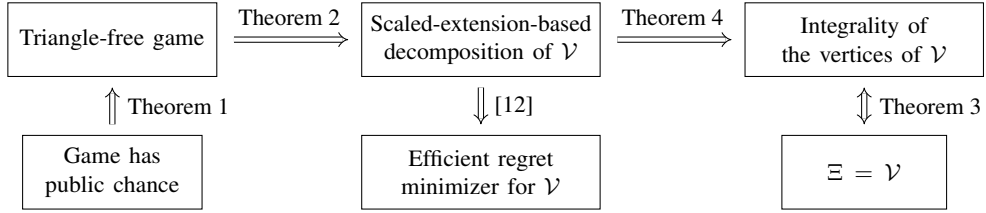

In \cref{fig:overview} we give an overview of the results in this paper and how
    they relate to each other. Our main result is that the polytope of correlation plans $\Xi$ coincides with the \vsfp{} $\cV$ when the game satisfies the \emph{triangle-freeness} condition that we introduce (\cref{def:triangle freeness}). As we show in \cref{thm:public chance implies triangle freeness}, every two-player game with public chance moves (which includes games with no chance moves at all) is triangle-free, but not all triangle-free games have public chance moves. So, our results also apply to some games where chance is not public.
The equality $\Xi = \cV$ in triangle-free games implies that an optimal EFCE, EFCCE and NFCCE can be computed in polynomial time. This is because $\cV$ has a polynomial (in the game size) description~\citep{Stengel08:Extensive} and the computation of an EFCE, EFCCE, NFCCE can be expressed as a linear program~\citep{Stengel08:Extensive,Farina20:Coarse}.

We prove $\Xi = \cV$ in several steps. First, we show that in triangle-free games, $\cV$ admits a structural decomposition in terms of \emph{scaled extension} operations. This type of decomposition of $\cV$ was introduced by \citet{Farina19:Efficient} as a way of ``unrolling'' the combinatorial structure of $\cV$ to construct an efficient regret minimization algorithm for $\Xi$ in two-player games without chance moves. We extend their construction to handle any triangle-free game. 
Then, we show a deep connection between the integrality of the vertices of the \vsfp{} $\cV$ and $\Xi$. Namely, in~\cref{thm:xi equal vonS}, we show that $\Xi = \cV$ holds if and only if all of $\cV$'s vertices have integer $\{0,1\}$ coordinates. Finally, in \cref{sec:integrality} we prove that $\cV$ has integral vertices by leveraging its structural decomposition.

%% file: text/preliminaries.tex
\section{Preliminaries}

\textbf{Extensive-form games}\quad \emph{Extensive-form games (EFGs)} are the standard model for games that are played on a game tree. EFGs can capture sequential and simultaneous moves as well as private information.
Each node in the EFG belongs to one player. One special player, called the \emph{chance player}, is used to model random stochastic events, such as rolling a die or drawing cards. In this paper, we only consider games that have two players in addition to potentially having a chance player.

Edges leaving from a node represent actions that a player can take at that node. To model private information, the game tree is supplemented with an \emph{information partition}, defined as a partition of nodes into sets called \emph{information sets}. Each node belongs to exactly one information set, and each information set is a nonempty set of tree nodes for the same Player $i$. An information set for Player $i$ denotes a collection of nodes that Player $i$ cannot distinguish among, given what she has observed so far. The symbols $\mathcal{I}_1$ and $\mathcal{I}_2$ denote the information partition of  Player 1 and 2, respectively.
Let $I_1$ and $I_2$ be information sets for Player 1 and 2, respectively. $I_1$ and $I_2$ are \emph{connected}, denoted $I_1 \conn I_2$, if there exist nodes $u \in I_1$ and $v \in I_2$ such that $u$ is on the path from the root to $v$, or vice versa.

We will only consider \emph{perfect-recall} games, that is, no player forgets what the player knew earlier. As a consequence, all nodes that belong to an information set $I$ share the same set of available actions (otherwise the player acting at those nodes would be able to distinguish among them), which we denote by $A_I$. We define the set of \emph{sequences} of Player $i$ as the set $\Sigma_i \defeq \{(I, a) : I \in \mathcal{I}_i, a \in A_I\} \cup \{\emptyseq\}$,
where the special element $\emptyseq$ is called \emph{empty sequence}. Given an information set $I \in \mathcal{I}_i$, we denote by $\sigma(I)$ the \emph{parent sequence} of $I$, defined as the last pair $(I', a') \in \Sigma_i$ encountered on the path from the root to any node $v \in I$; if no such pair exists 
we let $\sigma(I) = \emptyseq$.

An important concept in extensive-form correlation is \emph{relevance} of sequence pairs. Intuitively, two sequences are relevant if they belong to connected information sets or if either of them is the empty sequence. Formally, a pair of sequences $(\sigma_1, \sigma_2) \in \Sigma_1\times\Sigma_2$ is \emph{relevant}, denoted $\sigma_1\rele\sigma_2$, if either $\sigma_1$ or $\sigma_2$ or both is the empty sequence, or if $\sigma_1 = (I_1, a_1)$ and $\sigma_2 = (I_2, a_2)$ and $I_1 \conn I_2$. The set of all relevant sequence pairs is denoted $\Sigma_1\rele\Sigma_2$. Given $\sigma_1 \in \Sigma_1$ and $I_2 \in \mathcal{I}_2$, we say that $\sigma_1$ is relevant for $I_2$, and write $\sigma_1 \rele I_2$, if $\sigma_1 = \emptyseq$ or if $\sigma_1 = (I_1, a_1)$ and $I_1 \conn I_2$ (an analogous statement holds for $I_1 \rele \sigma_2$).
We say that a sequence $\sigma=(I,a)\in\Sigma_i$ for Player~$i$ is \emph{descendent} of another sequence $\sigma'=(I',a')\in\Sigma_i$ for the same player, denoted by $\sigma \succeq \sigma'$, if $\sigma=\sigma'$ or if there is a path from the root of the game to a node $v \in I$ that passes through action $a'$ at some node $v' \in I'$. We use the notation $\tau \succ \tau'$ to mean $\tau\succeq\tau' \land \tau \neq \tau'$.

A \emph{reduced-normal-form plan} $\pi_i$ for Player~$i$ defines a choice of action for every information set $I\in\infos{i}$ that is still reachable as a
result of the other choices in $\pi$ itself. We denote the set of
reduced-normal-form plans of Player~$i$ by $\Pi_i$. We denote by $\Pi_i(I)$ the subset of reduced-normal-form plans that prescribe all actions for Player~$i$ on the path from the root to information set $I$.

\textbf{Polytope of correlation plans ($\Xi$)}\quad
A correlated distribution $\mu$ over combinations of plans $\Pi_1 \times \Pi_2$ of the players can be thought of as a point in probability simplex $\Delta^{|\Pi_1 \times \Pi_2|}$. Since the number of plans of each player is exponential in the game tree size, so is that representation of $\mu$. Therefore, \citet{Stengel08:Extensive} introduced a more compact representation of $\mu$, called the \emph{correlation plan representation}. 
The set of all legal correlation plans
is denoted $\Xi$ and called the \emph{polytope of correlation plans}. The set $\Xi$ is a convex polytope in $\bbR_{\ge 0}^{|\Sigma_1\rele\Sigma_2|}$, so the number of variables is at most quadratic in the game tree size. However, it might still require an exponential number of \emph{constraints}.

An optimal EFCE, EFCCE, or NFCCE is an optimal correlation plan subject to a set of linear \emph{incentive constraints}~\citep{Stengel08:Extensive,Farina19:Correlation,Farina20:Coarse}. These constraints encode the requirement that the set of corrrelated behavior be \emph{incentive compatible} for the player, that is, that no player would be better off not following the recommended behavior than to always follow it. Hence, an optimal EFCE, EFCCE, or NFCCE can be computed as the solution of a linear program. Furthermore, the linear program can be solved in polynomial time if and only if $\Xi$ can be described with a polynomial number of linear constraints. Thus the characterization of the constraints that define $\Xi$ in various classes of games is important.

\textbf{The \vsfp{} ($\cV$)}\quad The characterization of the constraints that define $\Xi$ was initiated by~\citet{Stengel08:Extensive} in their landmark paper on extensive-form correlation. In particular, they show that in two-player perfect-recall games without chance moves, $\Xi$ coincides with a particular polytope $\cV$---which we call the \emph{\vsfp{}}---whose description only uses a polynomial number of linear constraints, which are ``probability-mass-conserving'' constraints:\!\!
\begin{equation}\label{eq:vonS constraints}\small
  \cV \defeq \left\{\vec{v} \in \bbR_{\ge 0}^{|\Sigma_1\rele\Sigma_2|}\!:\! \begin{array}{ll}
    \bullet\ \ v[\emptyseq, \emptyseq] = 1 \\[.8mm]
    \bullet\ \ \!\sum_{a \in A_I}~\!v[(I_1, a),\hspace{.1mm} \sigma_2] = v[\sigma(I_1),\hspace{.4mm} \sigma_2] \quad \forall I_1 \in \mathcal{I}_1, \sigma_2 \in \Sigma_2 \ \hspace{.0mm} \text{ s.t. } I_1 \rele \sigma_2\\[.8mm]
    \bullet\ \ \!\sum_{a \in A_J} v[\sigma_1, (I_2, a)] = v[\sigma_1, \sigma(I_2)] \quad \forall I_2 \in \mathcal{I}_2, \sigma_1 \in \Sigma_1 \,\text{ s.t. } \sigma_1 \rele I_2
  \end{array}\!\!\!\right\}\!.
\end{equation}
The polytope $\cV$ is well defined in every game. However, the equality $\Xi = \cV$ was known to hold only in two-player games without chance moves. In more general games, it is only known that $\Xi \subseteq \cV$. The main contribution of our paper is to show that the equality $\Xi = \cV$ holds in significantly more general games than two-player games without chance moves. We will isolate a condition, which we coin \emph{triangle freeness}, that is sufficient for $\Xi = \cV$ to hold. We also show that all two-player games where all chance moves are public (including two-player games without chance moves) are triangle free. 

%% file: text/decomposition.tex
\section{Scaled-Extension-Based Structural Decomposition for $\cV$}\label{sec:decomposition}

\citet{Farina19:Efficient} recently showed that in two-player games without chance moves, a particular structural decomposition theorem holds for the \vsfp{} $\cV$. At the core of their decomposition is a convexity-preserving operation, \emph{scaled extension}, defined as follows.
\begin{definition}[\citep{Farina19:Efficient}]\label{def:scaled extension}
Let $\cX$ and $\cY$ be nonempty, compact and convex sets, and let $h : \cX \to\bbR_{\ge 0}$ be a nonnegative affine real function. The \emph{scaled extension} of $\cX$ with $\cY$ via $h$ is defined as the set
\[
  \cX \ext^h \cY \defeq \{(\vec{x}, \vec{y}) : \vec{x} \in \cX,\ \vec{y} \in h(\vec{x}) \cY\}.
\]
\end{definition}
Specifically, they show that in two-player games without chance moves, $\cV$ admits a decomposition of the form $\cV = \{1\} \ext^{h_1} \cX_1 \ext^{h_2} \cX_2 \ext^{h_3} \cdots \ext^{h_n} \cX_n$, where each of the sets $\cX_i$ is either the singleton set $\{1\}$, or a probability simplex $\Delta^{s_i} \defeq \{\vec{x} \in \bbR^{s_i}_{\ge 0}: \|\vec{x}\|_1 = 1\}$ for some appropriate dimension $s_i$.

In this section, we significantly extend their result. As we will show, an analogous scaled-extension-based decomposition of $\cV$ exists in far more general games than those without chance moves. In particular, in \cref{sec:triangle freeness} we isolate a condition on the information structure of the game---which we coin \emph{triangle freeness}---that guarantees existence of a scaled-extension-based decomposition. Then, we will present an algorithm for computing such a decomposition, that is, finding the $h_i$ functions and sets $\cX_i$.
Since our full algorithm is rather intricate, we start by giving three examples of increasing complexity that capture the main intuition behind our structural decomposition routine.

%


        \begin{figure}[t]\centering
          \begin{tabular}{l|c|c|c}
            \toprule
            \ \rotatebox{90}{~Game tree}&
            \scalebox{.9}{\begin{tikzpicture}[baseline=-2.3cm]
              \def\dcha{1.0}
              \def\done{.45}
              \def\dtwo{.25}
              \node[chanode] (A) at (0, -.2) {};
              \node[pl1node] (B) at (-\dcha,-.8) {};
              \node[pl1node] (C) at (\dcha,-.8) {};
              \node[pl2node] (D) at (-\dcha-\done,-1.6) {};
              \node[pl2node] (E) at (-\dcha+\done,-1.6) {};
              \node[pl2node] (F) at (\dcha-\done,-1.6) {};
              \node[pl2node] (G) at (\dcha+\done,-1.6) {};
              \node[termina] (l1) at (-\dcha-\done-\dtwo,-2.4) {};
              \node[termina] (l2) at (-\dcha-\done+\dtwo,-2.4) {};
              \node[termina] (l3) at (-\dcha+\done-\dtwo,-2.4) {};
              \node[termina] (l4) at (-\dcha+\done+\dtwo,-2.4) {};
              \node[termina] (l5) at (\dcha-\done-\dtwo,-2.4) {};
              \node[termina] (l6) at (\dcha-\done+\dtwo,-2.4) {};
              \node[termina] (l7) at (\dcha+\done-\dtwo,-2.4) {};
              \node[termina] (l8) at (\dcha+\done+\dtwo,-2.4) {};

              \draw[semithick] (A) -- (B)
                                   --node[fill=white,inner sep=.9] {\scriptsize$1$} (D)
                                   --node[fill=white,inner sep=.9] {\scriptsize$1$} (l1);
              \draw[semithick] (B) --node[fill=white,inner sep=.9] {\scriptsize$2$} (E)
                                   --node[fill=white,inner sep=.9] {\scriptsize$1$} (l3);
              \draw[semithick] (D) --node[fill=white,inner sep=.9] {\scriptsize$2$} (l2);
              \draw[semithick] (E) --node[fill=white,inner sep=.9] {\scriptsize$2$} (l4);
              \draw[semithick] (A) -- (C)
                                   --node[fill=white,inner sep=.9] {\scriptsize$3$} (F)
                                   --node[fill=white,inner sep=.9] {\scriptsize$1$} (l5);
              \draw[semithick] (C) --node[fill=white,inner sep=.9] {\scriptsize$4$} (G)
                                   --node[fill=white,inner sep=.9] {\scriptsize$1$} (l7);
              \draw[semithick] (F) --node[fill=white,inner sep=.9] {\scriptsize$2$} (l6);
              \draw[semithick] (G) --node[fill=white,inner sep=.9] {\scriptsize$2$} (l8);

              \draw[black!60!white] (B) circle (.2);
              \node[black!60!white]  at ($(B) + (-.4, 0)$) {\textsc{a}};

              \draw[black!60!white] (C) circle (.2);
              \node[black!60!white]  at ($(C) + (.4, 0)$) {\textsc{b}};

              \draw[black!60!white] ($(D) + (0, .2)$) arc (90:270:.2);
              \draw[black!60!white] ($(D) + (0, .2)$) -- ($(G) + (0, .2)$);
              \draw[black!60!white] ($(D) + (0, -.2)$) -- ($(G) + (0, -.2)$);
              \draw[black!60!white] ($(G) + (0, -.2)$) arc (-90:90:.2);

              \node[black!60!white]  at ($(G) + (.4, 0)$) {\textsc{c}};
            \end{tikzpicture}}
             &
            \scalebox{.9}{\begin{tikzpicture}[baseline=-2.3cm]
              \def\dcha{1.0}
              \def\done{.45}
              \def\dtwo{.25}
              \node[chanode] (A) at (0, -.2) {};
              \node[pl1node] (B) at (-\dcha,-.8) {};
              \node[pl1node] (C) at (\dcha,-.8) {};
              \node[pl2node] (D) at (-\dcha-\done,-1.6) {};
              \node[pl2node] (E) at (-\dcha+\done,-1.6) {};
              \node[pl2node] (F) at (\dcha-\done,-1.6) {};
              \node[pl2node] (G) at (\dcha+\done,-1.6) {};
              \node[termina] (l1) at (-\dcha-\done-\dtwo,-2.4) {};
              \node[termina] (l2) at (-\dcha-\done+\dtwo,-2.4) {};
              \node[termina] (l3) at (-\dcha+\done-\dtwo,-2.4) {};
              \node[termina] (l4) at (-\dcha+\done+\dtwo,-2.4) {};
              \node[termina] (l5) at (\dcha-\done-\dtwo,-2.4) {};
              \node[termina] (l6) at (\dcha-\done+\dtwo,-2.4) {};
              \node[termina] (l7) at (\dcha+\done-\dtwo,-2.4) {};
              \node[termina] (l8) at (\dcha+\done+\dtwo,-2.4) {};

              \draw[semithick] (A) -- (B)
                                   --node[fill=white,inner sep=.9] {\scriptsize$1$} (D)
                                   --node[fill=white,inner sep=.9] {\scriptsize$1$} (l1);
              \draw[semithick] (B) --node[fill=white,inner sep=.9] {\scriptsize$2$} (E)
                                   --node[fill=white,inner sep=.9] {\scriptsize$1$} (l3);
              \draw[semithick] (D) --node[fill=white,inner sep=.9] {\scriptsize$2$} (l2);
              \draw[semithick] (E) --node[fill=white,inner sep=.9] {\scriptsize$2$} (l4);
              \draw[semithick] (A) -- (C)
                                   --node[fill=white,inner sep=.9] {\scriptsize$3$} (F)
                                   --node[fill=white,inner sep=.9] {\scriptsize$3$} (l5);
              \draw[semithick] (C) --node[fill=white,inner sep=.9] {\scriptsize$4$} (G)
                                   --node[fill=white,inner sep=.9] {\scriptsize$3$} (l7);
              \draw[semithick] (F) --node[fill=white,inner sep=.9] {\scriptsize$4$} (l6);
              \draw[semithick] (G) --node[fill=white,inner sep=.9] {\scriptsize$4$} (l8);

              \draw[black!60!white] (B) circle (.2);
              \node[black!60!white]  at ($(B) + (-.4, 0)$) {\textsc{a}};

              \draw[black!60!white] (C) circle (.2);
              \node[black!60!white]  at ($(C) + (.4, 0)$) {\textsc{b}};

              \draw[black!60!white] ($(D) + (0, .2)$) arc (90:270:.2);
              \draw[black!60!white] ($(D) + (0, .2)$) -- ($(E) + (0, .2)$);
              \draw[black!60!white] ($(D) + (0, -.2)$) -- ($(E) + (0, -.2)$);
              \draw[black!60!white] ($(E) + (0, -.2)$) arc (-90:90:.2);
              \node[black!60!white]  at ($(D) + (-.4, 0)$) {\textsc{c}};

              \draw[black!60!white] ($(F) + (0, .2)$) arc (90:270:.2);
              \draw[black!60!white] ($(F) + (0, .2)$) -- ($(G) + (0, .2)$);
              \draw[black!60!white] ($(F) + (0, -.2)$) -- ($(G) + (0, -.2)$);
              \draw[black!60!white] ($(G) + (0, -.2)$) arc (-90:90:.2);
              \node[black!60!white]  at ($(G) + (.4, 0)$) {\textsc{d}};
            \end{tikzpicture}}
            &
            \scalebox{.9}{\begin{tikzpicture}[baseline=-2.3cm]
              \def\dcha{1.0}
              \def\done{.45}
              \def\dtwo{.25}
              \node[chanode] (A) at (0, -.2) {};
              \node[pl1node] (B) at (-\dcha,-.8) {};
              \node[pl1node] (C) at (\dcha,-.8) {};
              \node[pl2node] (D) at (-\dcha-\done,-1.4) {};
              \node[pl2node] (E) at (-\dcha+\done,-1.8) {};
              \node[pl2node] (F) at (\dcha-\done,-1.4) {};
              \node[pl2node] (G) at (\dcha+\done,-1.8) {};
              \node[termina] (l1) at (-\dcha-\done-\dtwo,-2.4) {};
              \node[termina] (l2) at (-\dcha-\done+\dtwo,-2.4) {};
              \node[termina] (l3) at (-\dcha+\done-\dtwo,-2.4) {};
              \node[termina] (l4) at (-\dcha+\done+\dtwo,-2.4) {};
              \node[termina] (l5) at (\dcha-\done-\dtwo,-2.4) {};
              \node[termina] (l6) at (\dcha-\done+\dtwo,-2.4) {};
              \node[termina] (l7) at (\dcha+\done-\dtwo,-2.4) {};
              \node[termina] (l8) at (\dcha+\done+\dtwo,-2.4) {};

              \draw[black!60!white] ($(D) + (0, .15)$) arc (90:270:.15);
              \draw[black!60!white] ($(D) + (0, .15)$) -- ($(F) + (0, .15)$);
              \draw[black!60!white] ($(D) + (0, -.15)$) -- ($(F) + (0, -.15)$);
              \draw[black!60!white] ($(F) + (0, -.15)$) arc (-90:90:.15);
              \node[black!60!white]  at ($(D) + (-.4, 0)$) {\textsc{c}};

              \draw[black!60!white] ($(E) + (0, .15)$) arc (90:270:.15);
              \draw[black!60!white] ($(E) + (0, .15)$) -- ($(G) + (0, .15)$);
              \draw[black!60!white] ($(E) + (0, -.15)$) -- ($(G) + (0, -.15)$);
              \draw[black!60!white] ($(G) + (0, -.15)$) arc (-90:90:.15);
              \node[black!60!white]  at ($(G) + (.4, 0)$) {\textsc{d}};

              \draw[semithick] (A) -- (B)
                                   --node[fill=white,inner sep=.9] {\scriptsize$1$} (D)
                                   --node[fill=white,inner sep=.9,pos=.55] {\scriptsize$1$} (l1);
              \draw[semithick] (B) --node[fill=white,inner sep=.9] {\scriptsize$2$} (E)
                                   --node[fill=white,inner sep=.9] {\scriptsize$3$} (l3);
              \draw[semithick] (D) --node[fill=white,inner sep=.9,pos=.55] {\scriptsize$2$} (l2);
              \draw[semithick] (E) --node[fill=white,inner sep=.9] {\scriptsize$4$} (l4);
              \draw[semithick] (A) -- (C)
                                   --node[fill=white,inner sep=.9] {\scriptsize$3$} (F)
                                   --node[fill=white,inner sep=.9,pos=.55] {\scriptsize$1$} (l5);
              \draw[semithick] (C) --node[fill=white,inner sep=.9] {\scriptsize$4$} (G)
                                   --node[fill=white,inner sep=.9] {\scriptsize$3$} (l7);
              \draw[semithick] (F) --node[fill=white,inner sep=.9,pos=.55] {\scriptsize$2$} (l6);
              \draw[semithick] (G) --node[fill=white,inner sep=.9] {\scriptsize$4$} (l8);

              \draw[black!60!white] (B) circle (.2);
              \node[black!60!white]  at ($(B) + (-.4, 0)$) {\textsc{a}};

              \draw[black!60!white] (C) circle (.2);
              \node[black!60!white]  at ($(C) + (.4, 0)$) {\textsc{b}};
            \end{tikzpicture}}
            \\[3mm]
            \midrule
            \rotatebox{90}{\small Correlation plan}
            \rotatebox{90}{\small\quad fill-in order}
            &
            \scalebox{.9}{\begin{tikzpicture}[baseline=-2.3cm]
              \def\sep{.1}%
              \def\side{.42}%
              \def\dd{.07}%
              \tikzstyle{outer}=[semithick];
              \begin{scope}
                \draw[outer] (0, 0) rectangle (\side, -\side);
                \draw[outer] ($(\side + \sep,0)$) rectangle ($(3 * \side + \sep, -\side)$);
                \draw[outer] ($(0, -\side - \sep)$) rectangle ($(\side, -3 * \side - \sep)$);
                \draw[outer] ($(\side + \sep, -\side - \sep)$) rectangle ($(3 * \side + \sep, -3 * \side - \sep)$);

                \draw[outer] ($(0, -3*\side - 2*\sep)$) rectangle ($(\side, -5 * \side - 2*\sep)$);
                \draw[outer] ($(\side + \sep, -3*\side - 2*\sep)$) rectangle ($(3 * \side + \sep, -5 * \side - 2 * \sep)$);

                \draw ($(0, -2*\side -\sep)$) -- ($(\side,-2*\side-\sep)$);
                \draw ($(\side+\sep, -2*\side -\sep)$) -- ($(3*\side+\sep,-2*\side-\sep)$);
                \draw ($(2*\side+\sep, 0)$) -- ($(2*\side+\sep,-\side)$);
                \draw ($(2*\side+\sep, -\side-\sep)$) -- ($(2*\side+\sep,-3*\side-\sep)$);

                \draw ($(0, -4*\side -2*\sep)$) -- ($(\side,-4*\side-2*\sep)$);
                \draw ($(\side+\sep, -4*\side -2*\sep)$) -- ($(3*\side+\sep,-4*\side-2*\sep)$);

                \draw ($(2*\side+\sep, -3*\side - 2*\sep)$) -- ($(2*\side+\sep,-5*\side - 2*\sep)$);

                \node at ($(.5*\side,.2)$) {\scriptsize$\emptyseq$};
                \node at ($(\sep+1.5*\side,.2)$) {\scriptsize$1$};
                \node at ($(\sep+2.5*\side,.2)$) {\scriptsize$2$};

                \node at ($(-.2,-.5*\side)$) {\scriptsize$\emptyseq$};
                \node at ($(-.2,-\sep-1.5*\side)$) {\scriptsize$1$};
                \node at ($(-.2,-\sep-2.5*\side)$) {\scriptsize$2$};
                \node at ($(-.2,-2*\sep-3.5*\side)$) {\scriptsize$3$};
                \node at ($(-.2,-2*\sep-4.5*\side)$) {\scriptsize$4$};

                \fill[black!10!white,opacity=.9] ($(\dd,-\dd)$) rectangle node[black!60!white,opacity=1]{\circled{1}} ($(\side-\dd,-\side+\dd)$);

                \fill[black!10!white,opacity=.9] ($(\dd,-\side-\sep-\dd)$) rectangle node[black!60!white,opacity=1]{\circled{4}} ($(\side-\dd,-3*\side-\sep+\dd)$);
                \fill[black!10!white,opacity=.9] ($(\dd,-3*\side-2*\sep-\dd)$) rectangle node[black!60!white,opacity=1]{\circled{4}} ($(\side-\dd,-5*\side-2*\sep+\dd)$);

                \fill[black!10!white,opacity=.9] ($(\side+\sep+\dd,-\side-\sep-\dd)$) rectangle node[black!60!white,opacity=1]{\circled{3}} ($(3*\side+\sep-\dd,-3*\side-\sep+\dd)$);

                \fill[black!10!white,opacity=.9] ($(\side+\sep+\dd,-3*\side-2*\sep-\dd)$) rectangle node[black!60!white,opacity=1]{\circled{3}} ($(3*\side+\sep-\dd,-5*\side-2*\sep+\dd)$);

                \fill[black!10!white,opacity=.9] ($(\side+\sep+\dd,-\dd)$) rectangle node[black!60!white,opacity=1]{\circled{2}} ($(3*\side+\sep-\dd,-\side+\dd)$);
              \end{scope}
            \end{tikzpicture}}
            &
            \scalebox{.9}{\begin{tikzpicture}[baseline=-2.3cm]
              \def\sep{.1}%
              \def\side{.42}%
              \def\dd{.07}%
              \tikzstyle{outer}=[semithick];
              \begin{scope}
                \draw[outer] (0, 0) rectangle (\side, -\side);
                \draw[outer] ($(\side + \sep,0)$) rectangle ($(3 * \side + \sep, -\side)$);
                \draw[outer] ($(3 * \side + 2 * \sep, 0)$) rectangle ($(5 * \side + 2 * \sep, -\side)$);
                \draw[outer] ($(0, -\side - \sep)$) rectangle ($(\side, -3 * \side - \sep)$);
                \draw[outer] ($(\side + \sep, -\side - \sep)$) rectangle ($(3 * \side + \sep, -3 * \side - \sep)$);

                \draw[outer] ($(0, -3*\side - 2*\sep)$) rectangle ($(\side, -5 * \side - 2*\sep)$);
                \draw[outer] ($(3 * \side + 2 * \sep, -3*\side - 2*\sep)$) rectangle ($(5 * \side + 2 * \sep, -5 * \side - 2 * \sep)$);

                \draw ($(0, -2*\side -\sep)$) -- ($(\side,-2*\side-\sep)$);
                \draw ($(\side+\sep, -2*\side -\sep)$) -- ($(3*\side+\sep,-2*\side-\sep)$);

                \draw ($(2*\side+\sep, 0)$) -- ($(2*\side+\sep,-\side)$);
                \draw ($(2*\side+\sep, -\side-\sep)$) -- ($(2*\side+\sep,-3*\side-\sep)$);
                \draw ($(4*\side+2*\sep, 0)$) -- ($(4*\side+2*\sep,-\side)$);

                \draw ($(0, -4*\side -2*\sep)$) -- ($(\side,-4*\side-2*\sep)$);
                \draw ($(3*\side+2*\sep, -4*\side - 2*\sep)$) -- ($(5*\side+2*\sep,-4*\side-2*\sep)$);

                \draw ($(4*\side+2*\sep, -3*\side - 2*\sep)$) -- ($(4*\side+2*\sep,-5*\side - 2*\sep)$);

                \node at ($(.5*\side,.2)$) {\scriptsize$\emptyseq$};
                \node at ($(\sep+1.5*\side,.2)$) {\scriptsize$1$};
                \node at ($(\sep+2.5*\side,.2)$) {\scriptsize$2$};
                \node at ($(2*\sep+3.5*\side,.2)$) {\scriptsize$3$};
                \node at ($(2*\sep+4.5*\side,.2)$) {\scriptsize$4$};

                \node at ($(-.2,-.5*\side)$) {\scriptsize$\emptyseq$};
                \node at ($(-.2,-\sep-1.5*\side)$) {\scriptsize$1$};
                \node at ($(-.2,-\sep-2.5*\side)$) {\scriptsize$2$};
                \node at ($(-.2,-2*\sep-3.5*\side)$) {\scriptsize$3$};
                \node at ($(-.2,-2*\sep-4.5*\side)$) {\scriptsize$4$};

                \fill[black!10!white,opacity=.9] ($(\dd,-\dd)$) rectangle node[black!60!white,opacity=1]{\circled{1}} ($(\side-\dd,-\side+\dd)$);

                \fill[black!10!white,opacity=.9] ($(\dd,-\side-\sep-\dd)$) rectangle node[black!60!white,opacity=1]{\circled{2}} ($(\side-\dd,-3*\side-\sep+\dd)$);
                \fill[black!10!white,opacity=.9] ($(\dd,-3*\side-2*\sep-\dd)$) rectangle node[black!60!white,opacity=1]{\circled{2}} ($(\side-\dd,-5*\side-2*\sep+\dd)$);

                \fill[black!10!white,opacity=.9] ($(\side+\sep+\dd,-\side-\sep-\dd)$) rectangle node[black!60!white,opacity=1]{\circled{3}} ($(3*\side+\sep-\dd,-3*\side-\sep+\dd)$);

                \fill[black!10!white,opacity=.9] ($(3*\side+2*\sep+\dd,-3*\side-2*\sep-\dd)$) rectangle node[black!60!white,opacity=1]{\circled{3}} ($(5*\side+2*\sep-\dd,-5*\side-2*\sep+\dd)$);

                \fill[black!10!white,opacity=.9] ($(\side+\sep+\dd,-\dd)$) rectangle node[black!60!white,opacity=1]{\circled{4}} ($(3*\side+\sep-\dd,-\side+\dd)$);
                \fill[black!10!white,opacity=.9] ($(3*\side+2*\sep+\dd,-\dd)$) rectangle node[black!60!white,opacity=1]{\circled{4}} ($(5*\side+2*\sep-\dd,-\side+\dd)$);
              \end{scope}
            \end{tikzpicture}}
            &
            \scalebox{.9}{\begin{tikzpicture}[baseline=-2.3cm]
              \def\sep{.1}%
              \def\side{.42}%
              \def\dd{.07}%
              \tikzstyle{outer}=[semithick];
              \begin{scope}
                \draw[outer] (0, 0) rectangle (\side, -\side);
                \draw[outer] ($(\side + \sep,0)$) rectangle ($(3 * \side + \sep, -\side)$);
                \draw[outer] ($(3 * \side + 2 * \sep, 0)$) rectangle ($(5 * \side + 2 * \sep, -\side)$);
                \draw[outer] ($(0, -\side - \sep)$) rectangle ($(\side, -3 * \side - \sep)$);
                \draw[outer] ($(\side + \sep, -\side - \sep)$) rectangle ($(3 * \side + \sep, -3 * \side - \sep)$);
                \draw[outer] ($(3 * \side + 2 * \sep, -\side - \sep)$) rectangle ($(5 * \side + 2 * \sep, -3 * \side - \sep)$);

                \draw[outer] ($(0, -3*\side - 2*\sep)$) rectangle ($(\side, -5 * \side - 2*\sep)$);
                \draw[outer] ($(\side + \sep, -3*\side - 2*\sep)$) rectangle ($(3 * \side + \sep, -5 * \side - 2 * \sep)$);
                \draw[outer] ($(3 * \side + 2 * \sep, -3*\side - 2*\sep)$) rectangle ($(5 * \side + 2 * \sep, -5 * \side - 2 * \sep)$);

                \draw ($(0, -2*\side -\sep)$) -- ($(\side,-2*\side-\sep)$);
                \draw ($(\side+\sep, -2*\side -\sep)$) -- ($(3*\side+\sep,-2*\side-\sep)$);
                \draw ($(3*\side+2*\sep, -2*\side -\sep)$) -- ($(5*\side+2*\sep,-2*\side-\sep)$);
                \draw ($(2*\side+\sep, 0)$) -- ($(2*\side+\sep,-\side)$);
                \draw ($(2*\side+\sep, -\side-\sep)$) -- ($(2*\side+\sep,-3*\side-\sep)$);
                \draw ($(4*\side+2*\sep, 0)$) -- ($(4*\side+2*\sep,-\side)$);
                \draw ($(4*\side+2*\sep, -\side-\sep)$) -- ($(4*\side+2*\sep,-3*\side-\sep)$);

                \draw ($(0, -4*\side -2*\sep)$) -- ($(\side,-4*\side-2*\sep)$);
                \draw ($(\side+\sep, -4*\side -2*\sep)$) -- ($(3*\side+\sep,-4*\side-2*\sep)$);
                \draw ($(3*\side+2*\sep, -4*\side - 2*\sep)$) -- ($(5*\side+2*\sep,-4*\side-2*\sep)$);

                \draw ($(2*\side+\sep, -3*\side - 2*\sep)$) -- ($(2*\side+\sep,-5*\side - 2*\sep)$);
                \draw ($(4*\side+2*\sep, -3*\side - 2*\sep)$) -- ($(4*\side+2*\sep,-5*\side - 2*\sep)$);

                \node at ($(.5*\side,.2)$) {\scriptsize$\emptyseq$};
                \node at ($(\sep+1.5*\side,.2)$) {\scriptsize$1$};
                \node at ($(\sep+2.5*\side,.2)$) {\scriptsize$2$};
                \node at ($(2*\sep+3.5*\side,.2)$) {\scriptsize$3$};
                \node at ($(2*\sep+4.5*\side,.2)$) {\scriptsize$4$};

                \node at ($(-.2,-.5*\side)$) {\scriptsize$\emptyseq$};
                \node at ($(-.2,-\sep-1.5*\side)$) {\scriptsize$1$};
                \node at ($(-.2,-\sep-2.5*\side)$) {\scriptsize$2$};
                \node at ($(-.2,-2*\sep-3.5*\side)$) {\scriptsize$3$};
                \node at ($(-.2,-2*\sep-4.5*\side)$) {\scriptsize$4$};

                \fill[black!10!white,opacity=.9] ($(\dd,-\dd)$) rectangle node[black!60!white,opacity=1]{} ($(\side-\dd,-\side+\dd)$);

                \fill[black!10!white,opacity=.9] ($(\dd,-\side-\sep-\dd)$) rectangle node[black!60!white,opacity=1]{} ($(\side-\dd,-3*\side-\sep+\dd)$);
                \fill[black!10!white,opacity=.9] ($(\dd,-3*\side-2*\sep-\dd)$) rectangle node[black!60!white,opacity=1]{} ($(\side-\dd,-5*\side-2*\sep+\dd)$);

                \fill[black!10!white,opacity=.9] ($(\side+\sep+\dd,-\side-\sep-\dd)$) rectangle node[black!60!white,opacity=1]{} ($(3*\side+\sep-\dd,-3*\side-\sep+\dd)$);
                \fill[black!10!white,opacity=.9] ($(3*\side+2*\sep+\dd,-\side-\sep-\dd)$) rectangle node[black!60!white,opacity=1]{} ($(5*\side+2*\sep-\dd,-3*\side-\sep+\dd)$);

                \fill[black!10!white,opacity=.9] ($(\side+\sep+\dd,-3*\side-2*\sep-\dd)$) rectangle node[black!60!white,opacity=1]{} ($(3*\side+\sep-\dd,-5*\side-2*\sep+\dd)$);
                \fill[black!10!white,opacity=.9] ($(3*\side+2*\sep+\dd,-3*\side-2*\sep-\dd)$) rectangle node[black!60!white,opacity=1]{} ($(5*\side+2*\sep-\dd,-5*\side-2*\sep+\dd)$);

                \fill[black!10!white,opacity=.9] ($(\side+\sep+\dd,-\dd)$) rectangle node[black!60!white,opacity=1]{} ($(3*\side+\sep-\dd,-\side+\dd)$);
                \fill[black!10!white,opacity=.9] ($(3*\side+2*\sep+\dd,-\dd)$) rectangle node[black!60!white,opacity=1]{} ($(5*\side+2*\sep-\dd,-\side+\dd)$);
              \end{scope}
            \end{tikzpicture}}
            \\
            \bottomrule
          \end{tabular}
          \caption{Three examples of extensive-form games with increasingly complex information partitions. Crossed nodes belong the chance player, black round nodes belong to Player~1, white round nodes belong to Player~2, gray round sets define information sets, white squares denote terminal states. The numbers along the edges define concise names for sequences.}
          \label{fig:examples}
        \end{figure}
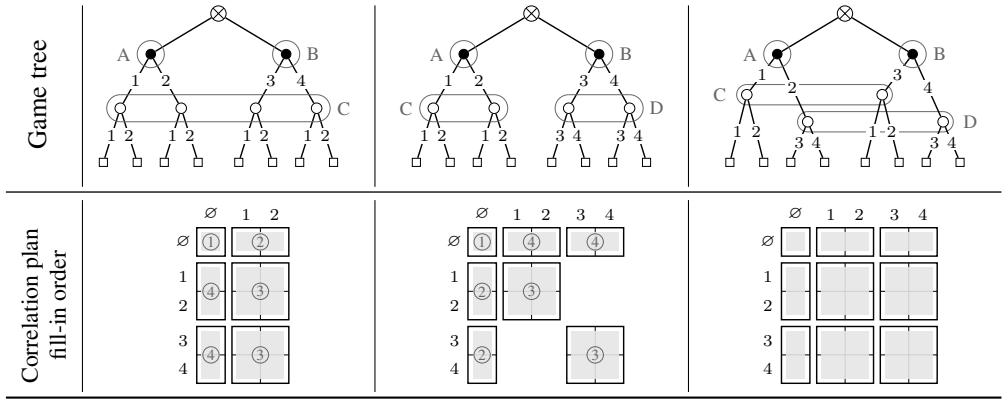

        \textbf{First example}\quad The first example is shown in the first column of \cref{fig:examples}. The game starts with a chance node, where two outcomes (say, heads or tails) are possible. After observing the outcome of the chance node, Player 1 chooses between two actions (say, the ``left'' and the ``right'' action). The choice as to whether to play the left or the right action can be different based on the observed chance outcome. After Player 1 has played their action, Player 2 has to pick whether to play their left or right action---however, Player 2 does not observe the chance outcome nor Player 1's action. The chance outcome is not observed by Player 2, so, this is not a public-chance game.

        The only information set \textsc{c} for Player 2 is connected to both information sets (denoted \textsc{a} and \textsc{b} in \cref{fig:examples}) of Player 1, so, all sequence pairs $(\sigma_1, \sigma_2) \in \Sigma_1\times\Sigma_2$ are relevant. Since Player 2 only has one information set, it is easy to incrementally generate the \vsfp{}. First, the fixed value $1$ is assigned to $v[\emptyseq,\emptyseq]$ (step \circled{1} in the fill-in order). Then, this value is split arbitrarily into the two (non-negative) entries $v[\emptyseq,1],v[\emptyseq,2]$ so that $v[\emptyseq,1]+v[\emptyseq,2]=v[\emptyseq,\emptyseq]$ in accordance with the \vsfc{}s. This operation can be expressed using scaled extension as $\{(v[\emptyseq,\emptyseq],v[\emptyseq,1],v[\emptyseq,2])\} = \{1\} \ext^h \Delta^2$, where $h$ is the identity function (step \circled{2} in the fill-in order). Then, $v[\emptyseq,1]$ is further split into $v[1,1]+v[2,1]=v[\emptyseq,1]$ and $v[3,1]+v[4,1]=v[\emptyseq,1]$, while $v[\emptyseq,2]$ is split into $v[1,2]+v[2,2]=v[\emptyseq,2]$ and $v[3,2]+v[4,2]=v[\emptyseq,2]$ (step \circled{3} of the fill-in order). These operations can be expressed as scaled extensions with $\Delta^2$. Now that the eight entries $v[\sigma_1,\sigma_2]$ for $\sigma_1\in\{1,2,3,4\}, \sigma_2\in\{1,2\}$ have been filled out, we fill in $v[\sigma_1, \emptyseq]$ for all $\sigma_1 \in \{1,2,3,4\}$ in accordance with the \vsfc{} $v[\sigma_1, \emptyseq] = v[\sigma_1,1]+v[\sigma_1,2]$ (step \circled{4}). In this step, we are not splitting any values, but rather we are summing already-filled-in entries in $v$ to form new entries. Specifically, we can extend the set of partially-filled-in vectors $\vec{v} = (v[\emptyseq,\emptyseq],v[\emptyseq,1],v[\emptyseq,2],v[1,1],v[2,1],v[3,1],v[4,1],v[1,2],v[2,2],v[3,2],v[4,2])$ with the new entry $v[1,\emptyseq]$ by using the scaled extension operation $\{\vec{v}\}\ext^h \{1\}$ where $h$ is the (linear) function that extracts the sum $v[\sigma_1,1]+v[\sigma_1,2]$ from $\vec{v}$. By doing so, we have incrementally filled in all entries in $\vec{v}$. Furthermore, by construction, we have that all \vsfc{}s $v[\sigma_1,\emptyseq]=v[\sigma_1,1]+v[\sigma_1,2]$ ($\sigma_1 \in \{\emptyseq,1,2,3,4\}$) and $v[\emptyseq,\sigma_2] = v[1,\sigma_2] + v[2,\sigma_2] = v[3,\sigma_2] + v[4,\sigma_2]$ ($\sigma_2 \in \{1,2\}$) must hold. So, the only two \vsfc{}s that we have ignored and might potentially be violated are $v[\emptyseq,\emptyseq]= v[1,\emptyseq] + v[2,\emptyseq]$ and $v[\emptyseq,\emptyseq]=v[3,\emptyseq]+v[4,\emptyseq]$. This concern is quickly resolved by noting that those constraints are implied by the other ones that we satisfy. In particular, by construction we have $v[1,\emptyseq] + v[2,\emptyseq] = (v[1,1]+v[1,2])+(v[2,1]+v[2,2]) = (v[1,1] + v[2,1]) + (v[1,2] + v[2,2]) = v[\emptyseq,1] + v[\emptyseq,2] = v[\emptyseq,\emptyseq]$, and an analogous statement holds for $v[3,\emptyseq] + v[4,\emptyseq]$. So, all constraints hold and the scaled-extension-based decomposition is finished.

   \begin{remark}\label{rem:critical}
        An approach that would start by splitting $v[\emptyseq,\emptyseq]$ into $v[1,\emptyseq]+v[2,\emptyseq]=v[\emptyseq,\emptyseq]$ and $v[3,\emptyseq]+v[4,\emptyseq]=v[\emptyseq,\emptyseq]$, thereby inverting the order of fill-in steps \circled{4} and \circled{2}, would fail. Indeed, after filling $v[\sigma_1,\sigma_2]$ for all $\sigma_1 \in \{1,2,3,4\}, \sigma_2 \in \{1,2\}$), there would be no clear way of guaranteeing that $v[1,1]+v[2,1] = v[3,1]+v[4,1] \ (= v[\emptyseq,1])$.
    \end{remark}

    \textbf{Second example}\quad We now consider a variation of the game from the first example, where Player~2 observes the chance outcome but not the actions selected by Player~1. This game, shown in the middle column of \cref{fig:examples}, has \emph{public} chance moves, because the chance outcome is observed by all players. In this game, not all pairs of information sets are connected. In fact, only $(\textsc{a},\textsc{c})$ and $(\textsc{b},\textsc{d})$ are connected information set pairs. Correspondingly, entries such as $v[1,3]$, $v[4,2]$, and $v[2,4]$ are not defined in the correlation plans for the game. This observation is crucial, and will set apart this example from the next one. To fill in any correlation plan, we can start by splitting $v[\emptyseq,\emptyseq]$ into $v[1,\emptyseq]+v[2,\emptyseq]=v[\emptyseq,\emptyseq]$ and $v[3,\emptyseq]+v[4,\emptyseq]=v[\emptyseq,\emptyseq]$ (fill-in step \circled{2} in the figure). Both operations can be expressed as a scaled extension of partially-filled-in vectors with $\Delta^2$, scaled by the affine function that extracts $v[\emptyseq,\emptyseq]=1$ from the partially-filled-in correlation plans. Then, we further split those values into entries $v[\sigma_1,1]+v[\sigma_1,2] = v[\sigma_1,\emptyseq]$ for $\sigma_1\in\{1,2\}$ in accordance with the \vsfc{}. Similarly, we will in $v[\sigma_1,3],v[\sigma_1,4]$ for $\sigma_1\in\{3,4\}$ in accordance with the constraint $v[\sigma_1,3]+v[\sigma_1,4] = v[\sigma_1,\emptyseq]$ for $\sigma_1\in\{1,2\}$ (fill-in step \circled{3}). Finally, we recover the values of $v[\emptyseq, \sigma_2]$ for $\sigma_2 \in\{1,2,3,4\}$ with a scaled extension with the singleton set $\{1\}$ as discussed in the previous example. Again, it can be checked that despite the fact that we ignored the constraints $v[\emptyseq,1]+v[\emptyseq,2] = v[\emptyseq,\emptyseq]$ and $v[\emptyseq,3]+v[\emptyseq,4] = v[\emptyseq,\emptyseq]$, those constraints are automatically satisfied by constuction. In this case, we were able to sidestep the issue raised in \cref{rem:critical} because of the particular connection between the information sets.

    \textbf{Third example}\quad Finally, we propose a third example in the third column of~\cref{fig:examples}. It is a variation of the first example, where Player~2 now observes Player~1's action but \emph{not} the chance outcome. The most significant difference with the second example is that the information structure of the game is now such that all pairs of information sets of the players are connected. Hence, the problem raised in \cref{rem:critical} cannot be avoided. Our decomposition algorithm cannot handle this example.

        \subsection{A Sufficient Condition for the Existence of a Scaled-Extension-Based Decomposition}\label{sec:triangle freeness}

        The third example in the previous section highlights an unfavorable situation in which our decomposition attempt based on incremental generation of the correlation plan. In order to codify all situations in which that issue does not arise, we introduce the concept of rank of an information set.

\begin{definition}
    Let $i \in \{1,2\}$ be one player, and let $-i$ denote the other player. Furthermore, let $I \in \infos{i}$ and $\sigma \in \Sigma_{-i}$. The $\sigma$-rank of $I$ is the cardinality of the set $\{J \in \infos{-i} : J \conn I, \sigma(J) = \sigma \}$.
\end{definition}
\vspace{-2mm}

The issue in \cref{rem:critical} can be stated in terms of the ranks.
Consider a relevant sequence pair $(\sigma_1,\sigma_2)\in\Sigma_1\rele\Sigma_2$ and two connected information sets $I_1 \conn I_2$ such that $\sigma(I_1) = \sigma_1, \sigma(I_2) = \sigma_2$. If the $\sigma_1$-rank of $I_2$ and the $\sigma_2$-rank of $I_1$ are both greater than 1, the issue cannot be avoided and the decomposition will fail. For example, in the third example, where our decomposition fails, all information sets have  $\emptyseq$-rank 2.
We prove that such situations cannot occur, provided the game satisfies the following condition, which can be verified in polynomial time in the size of the EFG.
        \begin{definition}[Triangle-freeness]\label{def:triangle freeness}
            A two-player extensive-form game is \emph{triangle-free} if, for any choice of two distinct information sets $I_1, I_2 \in \infos{1}$ such that $\sigma(I_1) = \sigma(I_2) = \sigma_1$ and two distinct information sets $J_1, J_2 \in \infos{2}$ such that $\sigma(J_1) = \sigma(J_2) = \sigma_2$,
            it is never the case that $I_1\conn J_1 \land I_2 \conn J_2 \land I_1 \conn J_2$.
        \end{definition}\vspace{-2mm}

        In \cref{thm:public chance implies triangle freeness} we show that games with public chance (which includes games with no chance moves at all) always satisfy the triangle-freeness condition of \cref{def:triangle freeness}.

        \begin{restatable}{theorem}{thmpublicchance}\label{thm:public chance implies triangle freeness}
          A two-player extensive-form game with public chance moves is triangle-free.
        \end{restatable}\vspace{-2mm}

However, not all triangle-free games must have public chance nodes. For example, the topmost game in \cref{fig:examples} is triangle-free, but in that game the chance outcome is not public to Player 2. So, our results apply more broadly than games with public chance moves.

        \subsection{Computation of the Decomposition}\label{sec:computation}

We present our algorithm following the same structure as~\cite{Farina19:Efficient}. Like theirs, our algorithm consists of a recursive function, $\textsc{Decompose}$. It takes  three arguments: (i) a sequence pair $(\sigma_1, \sigma_2) \in \Sigma_1 \rele \Sigma_2$, (ii) a subset $\mathcal{S}$ of the set of all relevant sequence pairs, and (iii) a set $\mathcal{D}$ where only the entries indexed by the elements in $\mathcal{S}$ have been filled in. The decomposition for the whole \vsfp{} $\cV$ is computed by calling $\textsc{Decompose}((\emptyseq,\emptyseq), \{(\emptyseq,\emptyseq)\}, \{ (1) \})$---this corresponds to the starting situation in which only the entry $v[\emptyseq,\emptyseq]$ has been filled in (denoted as fill-in step \circled{1} in \cref{fig:examples}). Each call to \textsc{Decompose} returns a pair $(\mathcal{S}',\mathcal{D}')$ of updated indices and partial vectors, to reflect the new entries that were filled in during the call.

$\textsc{Decompose}((\sigma_1, \sigma_2), \mathcal{S}, \mathcal{D})$ operates as follows (we denote with $-i$ the opponent for Player~$i$):\vspace{-1mm}
\begin{enumerate}[leftmargin=*,nolistsep,itemsep=0mm]
  \item Let $\cJ_i \defeq \{I \in \infos{i} : I \rele \sigma_{-i}, \sigma(I) = \sigma_i\}$ for all $i \in \{1,2\}$, and $\cJ^* \gets \emptyset$.
  \item For each $(I_1, I_2) \in \cJ_1\times\cJ_2$ such that $I_1 \conn I_2$, if the $\sigma_{2}$-rank of $I_{1}$ is greater than or equal to the $\sigma_1$-rank of $I_2$, we \emph{branch} on Player~1, update $\cJ^* \gets \cJ^* \cup \{I_1\}$. Else, update $\cJ^* \gets \cJ^* \cup \{I_2\}$.
  \item\label{pt:branch 1} For each $i\in \{1,2\}$ and $I \in \cJ_i$ such that the $\sigma_{-i}$-rank of $I$ is 0, do $\cJ^* \gets \cJ^* \cup \{I\}$.
  \item\label{pt:branch 2} For each $I \in \cJ^*$: (Below we assume that $I \in \infos{1}$, the other case is symmetrical)
     \begin{enumerate}[label={(\alph*)},ref={\ref*{pt:branch 2}(\alph*)}]
       \item\label{pt:step2} Fill in all entries $\{v[(I,a),\sigma_2]:a\in A_{I}\}$ by splitting $v[\sigma_1,\sigma_2]$. This can be expressed using a scaled extension operation as $\mathcal{D} \gets \mathcal{D} \ext^h \symp{|A_{I}|}$ where $h$
                    extracts $\xi[\sigma_1,\sigma_2]$ from any partially-filled-in vector.
             \item Update $\mathcal{S} \gets \mathcal{S} \cup \{((I,a),\sigma_2)\}$ to reflect that the entries corresponding to $(I,a) \rele \sigma_2$ have been filled in.
             \item\label{pt:step3} For each $a\in A_{I}$ we assign $(\mathcal{S}, \mathcal{D}) \gets \textsc{Decompose}(((I, a), \sigma_2), \mathcal{S}, \mathcal{D})$.
             \item\label{pt:step4} Let $\cK \defeq \{J \in \infos{2}\!:\! I \conn J\}$. For all $J \in \infos{2}$ such that $\sigma(J) \succeq (J'\!, a')$ for some $J'\in\cK, a'\in A_{J'}$:
                        \begin{itemize}[nolistsep,itemsep=0mm]
                            \item If $I \conn J$, then for all $a \in A_J$ we fill in the sequence pair $\xi[\sigma_1, (J,a)]$ by assigning its value in accordance with the \vsfc{} $\xi[\sigma_1, (J,a)] = \sum_{a^* \in A_{I^*}} \xi[(I^*, a^*), (J,a)]$ via the scaled extension $\mathcal{D} \gets \mathcal{D} \ext^h \{ 1\}$ where the linear function $h$ maps a partially-filled-in vector to the value of $\sum_{a^* \in A_{I^*}} \xi[(I^*, a^*), (J,a)]$. 
                            \item Otherwise, we fill in the entries $\{\xi[\sigma_1, (J,a)] : a\in A_J\}$, by splitting the value $\xi[\sigma_1, \sigma(J)]$. In this case, we let
                            $\mathcal{D} \gets \mathcal{D} \ext^h \symp{|A_J|}$ where $h$ extracts the entry $\xi[\sigma_1, \sigma(J)]$ from a partially-filled-in vector in $\mathcal{D}$.
            \end{itemize}
     \end{enumerate}
     \item At this point, all the entries corresponding to indices $\tilde{\mathcal{S}} = \{(\sigma'_1, \sigma'_2): \sigma'_1 \succeq \sigma_1, \sigma'_2 \succeq \sigma_2\}$ have been filled in, and we return $(\mathcal{S} \cup \tilde{\mathcal{S}}, \mathcal{D})$.
\end{enumerate}
The above algorithm formalizes and generalizes the first two examples of \cref{fig:examples}. For example, step \circled{2} of the fill-in order in either example is captured in Step~\ref{pt:step2}, while fill-in step \circled{3} corresponds to Step~\ref{pt:step3}. Finally, fill-in step~\circled{4} corresponds to Step~\ref{pt:step4}.

Compared to the decomposition algorithm by \citet{Farina19:Efficient}, our branching steps (Step~\ref{pt:branch 2}) are significantly more intricate. This is because, compared to their setting (that is, two-player games without chance moves) where at least one player has at most one information set with rank strictly greater than one, we have to account for multiple information sets with rank greater than one. Since two-player games without chance moves are a special case of two-player games with public chance moves, our algorithm completely subsumes that of~\citet{Farina19:Efficient}.

A proof of correctness for the algorithm is in \cref{app:decomposition}. In particular, the following holds.
        \begin{restatable}{theorem}{thmdecomposition}\label{thm:decomposition}
          The \vsfp{} $\cV$ of a two-player perfect-recall triangle-free EFG can be expressed via a sequence of scaled extensions with simplexes and singleton sets:
\begin{align}
  \cV = \{1\} \ext^{h_1} \cX_1 \ext^{h_2} \cX_2 \ext^{h_3} \cdots \ext^{h_n} \cX_n, \text{ where, for } i = 1,\dots, n,\text{ either } \cX_i = \Delta^{s_i}\text{ or } \cX_i = \{1\},\label{eq:vonS decomposition}
\end{align}
and $h_i$ is a linear function.
 Furthermore, an exact algorithm exists to compute such expression in linear time in the dimensionality of $\cV$, and so, in time at most quadratic in the size of the game.
        \end{restatable}\vspace{-1mm} 

%% file: text/integrality_of_V.tex
\section{Bridging the Gap Between $\cV$ and $\Xi$}\label{sec:integrality}

As noted by~\citet{Stengel08:Extensive}, the inclusion $\Xi \subseteq \cV$ holds trivially in any game. The reverse inclusion, $\Xi \supseteq \cV$, was shown for two-player games without chance moves, but no complete characterization as to when that reverse inclusion holds was known before our paper. In \cref{thm:xi equal vonS}, we contribute a new connection between the reverse inclusion $\Xi \supseteq \cV$ and the integrality of the vertices of the \vsfp{} (all proofs are in \cref{app:xi equal vonS}).

\begin{restatable}{theorem}{thmvinxi}\label{thm:xi equal vonS}
  Let $\Gamma$ be a two-player perfect-recall extensive-form game, let $\cV$ be its \vsfp{}, and let $\Xi$ be its polytope of correlation plans. Then, $\Xi = \cV$ if and only if all vertices of $\cV$ have integer $\{0,1\}$ coordinates.
\end{restatable}




As it turns out, the scaled-based decomposition of $\cV$ can be used to conclude the integrality of the vertices of $\cV$, by leveraging the following analytical result about the scaled extension operation.

\begin{restatable}{lemma}{lemvertscext}\label{lem:vertices of scext}
    Let $\cX, \cY$, and $h$ be as in \cref{def:scaled extension}. If $\cX$ is a convex polytope with vertices $\{\vec{x}_1, \dots, \vec{x}_n\}$, and $\cY$ is a convex polytope with vertices $\{\vec{y}_1, \dots, \vec{y}_m\}$, then $\cX \ext^h \cY$ is a convex polytope whose vertices are a nonempty subset of $\{(\vec{x}_i, h(\vec{x}_i)\vec{y}_j) : i \in \{1,\dots, n\}, j \in \{1,\dots, m\}\}$.
\end{restatable}

In particular, by applying \cref{lem:vertices of scext} inductively on the structure of the scaled-extension-based structural decomposition of $\cV$, we obtain the following theorem.

\begin{restatable}{theorem}{thmintegrality}\label{thm:integrality}
    Let $\cV$ be the \vsfp{} of a two-player triangle-free game (\cref{def:triangle freeness}). All vertices of $\cV$ have integer $\{0,1\}$ coordinates.
\end{restatable}

Finally, combining \cref{thm:integrality} and \cref{thm:xi equal vonS}, we obtain the central theorem of this paper.
\begin{restatable}{theorem}{thmmain}\label{thm:main}
  In a two-player perfect-recall extensive-form game that satisfies the triangle-freeness condition (\cref{def:triangle freeness}), the polytope of correlation plans coincides with the \vsfp{}. Consequently, an optimal EFCE, EFCCE, or NFCCE can be computed in polynomial time (in the size of the input extensive-form game) in two-player triangle-free games.
\end{restatable}

%
A consequence of $\cV = \Xi$ is that the linear programs for EFCE~\citep{Stengel08:Extensive}, EFCCE~\citep{Farina20:Coarse} and NFCCE~\citep{Farina20:Coarse}---originally formulated for two-player games without chance moves only---hold verbatim for any triangle-free game. So, an optimal EFCE, EFCCE, and NFCCE can be computed in polynomial time as the solution of those linear programs.
%
Furthermore, the scaled-extension-based decomposition for triangle-free games (\cref{sec:decomposition}) can be combined with the \emph{scaled extension regret circuit} introduced by~\citet{Farina19:Efficient,Farina19:Regret} to construct a scalable regret minimization algorithm for $\cV=\Xi$. That, in turn, can be used to compute an EFCE, EFCCE, and NFCCE in large triangle-free games that are too large for traditional linear programming methods.

%% file: text/experiments.tex
\section{Experimental Evaluation}
\vspace{-1mm}

We implemented the scaled-extension-based decomposition routine of \cref{sec:decomposition}. We test our decomposition algorithm for triangle-free games on Goofspiel~\citep{Ross71:Goofspiel}, a popular benchmark game in computational game theory. In Goofspiel, each player has a personal deck of cards made of $k$ different ranks (from $1$ to $k$). A third deck (the ``prize'' deck) is shuffled and put face down on the board at the beginning of the game. In each turn, the topmost card from the prize deck is \emph{publicly} revealed. Then, each player privately picks a card from their hand---this card acts as a bid to win the card that was just revealed from the prize deck. The player that bids the highest wins the prize card. We use an established tie-breaking rule: the prize card is discarded if the players' bids are equal. Furthermore, we adopt the convention that only the winner is revealed, but not the bids, in accordance with prior computational game theory literature~\cite{Lanctot09:Monte,Lanctot13:Monte}. The players’ scores are computed as the sum of the values of the prize cards they have won. Because of the tie-breaking rule, Goofspiel is a general-sum game. Furthermore, since all chance outcomes are public, it is a triangle-free game.

In \cref{fig:experiments}(left) we report the performance of our decomposition routine for $k=3,4,5$, both in terms of number of scaled extension operations required in the decomposition (\cref{thm:decomposition}) and of runtime of our single-threaded implementation, as well as the dimensions of the games. The runtime was averaged over 100 independent runs. Our decomposition algorithm performs well, and is able to scale to the largest game (Goofspiel with $k=5$ ranks, which has $3.6\times 10^7$ relevant sequence pairs). In \cref{fig:experiments}(right) we used the characterization $\Xi = \cV$ to compute the set of all payoffs that can be reached by an EFCE, EFCCE, or NFCCE in $3$-rank Goofspiel. The sets are highly non-trivial, and reinforce the observation that the behaviors and incentives that can be induced through extensive-form correlation are subtle and complex~\citep{Farina19:Correlation}. The sets of
reachable payoff vectors was computed using \emph{Polymake}, a library for
computational polyhedral geometry~\citep{Gawrilow00:Polymake,Assarf17:Computing}.

\begin{figure}[H]
    \vspace{-2mm}
    \begin{minipage}[b]{9.6cm}
        \begin{figure}[H]\small\centering
        \setlength{\tabcolsep}{.6mm}%
        \renewcommand{\arraystretch}{1.3}%
        \sisetup{
            scientific-notation=true,
            round-precision=1,
            round-mode=places,
            exponent-product={\mkern-4mu\times\mkern-4.5mu}
        }%
    \scalebox{.94}{\begin{tabular}{c|cc|cc|cc}
        \toprule
        \multirow{2}{*}{\begin{minipage}{1cm}\centering\bf Deck\\size\end{minipage}} & \multicolumn{2}{c|}{\bf Information sets} & \multicolumn{2}{c|}{\bf Sequences} & \multicolumn{2}{c}{\bf Decomposition} \\
        & $|\infos{1} \cup \infos{2}|$ & $|\infos{1}\conn\infos{2}|$ & $|\Sigma_1\cup\Sigma_2|$ & $|\Sigma_1\rele\Sigma_2|$ & Num $\ext^h$ & Runtime\\
        \midrule
        $3$ ranks & \num{426} & \num{1077} & \num{524} & \num{3262} & \num{2931} & 2ms  \\
        $4$ ranks & \num{17432} & \num{80884} & \num{21298} & \num{265393} & \num{235956} & 1.1s\\
        $5$ ranks & \num{1175330} & \num{10505585} & \num{1428452} & \num{36102736} & \num{31901355} & 43.8s  \\
        \bottomrule
    \end{tabular}}\vspace{-1mm}
    \caption{ (Left) Dimensions of games and runtime of decomposition algorithm (\cref{thm:decomposition}). (Right) Payoffs that can be reached using an EFCE, EFCCE, or NFCCE in $3$-rank Goofspiel.}\label{fig:experiments}
    \end{figure}
    \end{minipage}%
    \begin{minipage}[b]{3cm}
      \includegraphics[scale=.77]{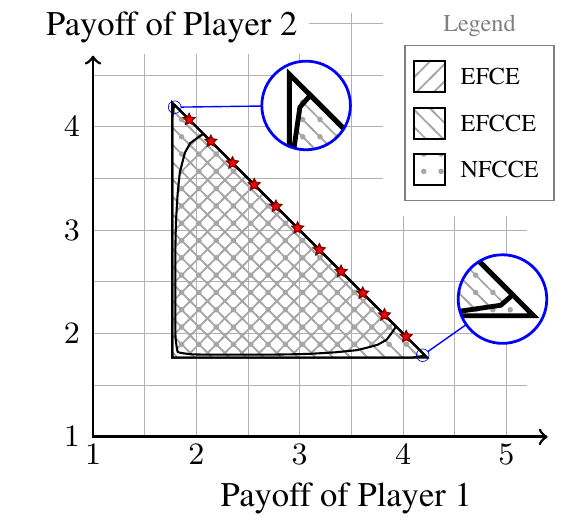}
    \end{minipage}
    \vspace{-2mm}
\end{figure}

We also implemented the linear programming formulation of EFCE described by~\citet{Stengel08:Extensive}, and the scalable regret minimization algorithm of~\citet{Farina19:Efficient}. We use the Gurobi commercial linear programming solver to solve the linear program formulation. As the game size increases, the barrier algorithm is the only algorithm that can solve the linear program. However, even that quickly becomes impractical. In the largest game, Gurobi uses roughly 200GB of memory, spends approximately 90 minutes to precondition the linear program, and requires slightly more than 20 minutes to perform each iteration of the barrier method using 30 threads. The regret minimization scales significantly more favorable in the large game. It requires roughly 6 seconds per iteration, and reaches $10^{-2}$ infeasibility in 4 minutes, $10^{-3}$ infeasibility in 12 minutes, and $10^{-4}$ infeasibility in 110 minutes. Additional data about the experiment is available in \cref{app:experiments}.

%% file: text/conclusions.tex
\section{Conclusions}
\vspace{-1mm}


We showed that an optimal extensive-form correlated equilibrium, extensive-form coarse correlated equilibrium, and normal-form coarse correlated equilibrium can be computed in polynomial time in two-player perfect-recall games that satisfy a certain \emph{triangle-freeness} condition that we introduced and that can be checked in polynomial time. To show that such equilibria can be found in polynomial time, we gave and combined several results that may be of independent interest: (1) the existence of a scaled-extension-based structural decomposition for the \vsfp{} of the game, (2) a characterization of when the \vsfp{} coincides with the polytope of correlation plans, and (3) a result about the integrality of the vertices of the \vsfp{} in triangle-free games. 

%% file: text/appendix_proofs.tex
\input{text/appendix_proofs_decomposition}
\input{text/appendix_proofs_integrality}

%% file: text/appendix_proofs_decomposition.tex
\section{Scaled-Extension-Based Structural Decomposition for $\cV$}\label{app:decomposition}

\subsection{Triangle-Freeness}

\begin{lemma}\label{lem:ranks}
    Consider a triangle-free game, let $(\sigma_1, \sigma_2) \in \Sigma_1 \rele \Sigma_2$, and let $I_1 \conn I_2$ be such that $\sigma(I_1) = \sigma_1, \sigma(I_2) = \sigma_2$. Then, at most one between the $\sigma_1$-rank of $I_2$ and the $\sigma_2$-rank of $I_1$ is strictly larger than $1$.
\end{lemma}
\begin{proof}
  The results follows almost immediately from the definition of triangle-freeness. We prove the statement by contradiction. Let $(\sigma_1, \sigma_2) \in \Sigma_1\rele\Sigma_2$ be a relevant sequence pair, and let information sets $I_1 \in \infos{1},I_2 \in \infos{2}$ be such that $\sigma(I_1)=\sigma_1, \sigma(I_2) = \sigma_2$. Furthermore, assume that the $\sigma_1$-rank of $I_2$ is greater than $1$, and at the same time the $\sigma_2$-rank of $I_1$ is greater than $1$. Since the $\sigma_2$-rank of $I_1$ is greater than $1$, there exists an information set $I'_2 \in \infos{2}, \sigma(I'_2) = \sigma_2$, distinct from $I_2$, such that $I_1 \conn I'_2$. Similarly, because the $\sigma_1$-rank of $I_2$ is greater than $1$, there exists an information set $I'_1 \in \infos{1},\sigma(I'_1) = \sigma_1$, distinct from $I_1$, such that $I'_1\conn I_2$. But then, we have found $I_1,I'_1\in\infos{1}$ and $I'_2,I_2\in\infos{2}$ such that $\sigma(I_1)=\sigma(I_2)=\sigma_1, \sigma(I'_2)=\sigma(I_2)=\sigma_2$ such that $I_1 \conn I'_2, I'_1 \conn I_2$, and $I_1\conn I_2$. So, the game is \emph{not} triangle-free, contradiction.
\end{proof}

\thmpublicchance*
\begin{proof}
  For contradiction, let $I_1, I_2$ be two distinct information sets for Player~1 such that $\sigma(I_1) = \sigma(I_2)$, let $J_1,J_2$ be two distinct information sets for Player~2 such that $\sigma(J_1) = \sigma(J_2)$, and assume that $I_1 \conn J_1, I_2 \conn J_2, I_1 \conn J_2$. By definition of connectedness, there exist nodes $u\in I_1, v \in J_1$ such that $v$ is on the path from the root to $u$, or \emph{vice versa}. Similarly, there exist nodes $u' \in I_2, v' \in J_2$ such that $u'$ is on the path from the root to $v'$, or \emph{vice versa}. Let $w$ be the lowest common ancestor of $u$ and $u'$. It is not possible that $w=u$ or $w = u'$, because otherwise the parent sequences of $I_1$ and $I_2$ would be different. So $w$ must be a strict ancestor of both $u$ and $u'$, and $u$ and $u'$ must be reached using \emph{different} edges at $w$. Therefore, node $w$ cannot belongs to Player~1, or otherwise it again would not be true that $\sigma(I_1) = \sigma(I_2)$. So, there are only two possible cases: either $w$ belongs to Player~2, or it belongs to the chance player. We break the analysis accordingly.
\begin{itemize}[nolistsep,itemsep=1mm,leftmargin=5mm]
  \item \textbf{First case:} $w$ belongs to Player~2. From above, we know that $u$ and $u'$ are reached by following different branches at $w$. So, if both $v$ and $v'$ were strict descendants of $w$, they would need to be on two different branches of $w$ (because they are connected to $u$ and $u'$ respectively), violating the condition $\sigma(J_1) = \sigma(J_2)$. So, at least one between $v$ and $v'$ is on the path from the root to $w$ (inclusive). But then either $v$ is an ancestor of $v'$, or \emph{vice versa}. Either case violates the hypothesis that $\sigma(J_1) = \sigma(J_2)$.
  \item \textbf{Second case:} $w$ belongs to the chance player. If any between $v$ and $v'$ is an ancestor of $w$, then necessarily either $v$ is an ancestor of $v'$, or $v'$ is an ancestor of $v$. Either case violates the condition $\sigma(J_1) = \sigma(J_2)$. So, both $v$ and $v'$ must be descendants of $w$. Because $v$ is on the path from the root to $u$ (or \emph{vice versa}), and $v'$ is on the path from the root to $u'$ (or \emph{vice versa}), then necessarily $u, v$ and $u',v'$ are on two different branches of the chance node $w$. To fix names, call $a$ the action at $w$ that must be taken to (eventually) reach $u$ and $v$, and let $b$ be the action that must be taken to (eventually) reach $u'$ and $v'$. Now, we use the hypothesis that $I_1 \conn J_2$, that is, there exists $u'' \in I_1, v'' \in J_2$ such that $u''$ is on the path from the root to $v''$ or \emph{vice versa}. Assume that $u''$ is on the path from the root to $v''$. Since $u''$ belongs to the same information set as $u$ (that is, $I_1$), and since chance is public by hypothesis, then Player~1, when acting at $u$ and $u''$, must have observed action $a$ at $w$. In other words, the path from the root to $u''$ must pass through action $a$ at $w$. But then, using the fact that $u''$ is on the path from the root to $v''$, this means that the path from the root to $v''$ passes through action $a$. However, the path from the root to $v'$ passes through action $b$. Since chance is public, nodes $v'$ and $v''$ cannot be in the same information set, because Player~2 is able to distinguish them by means of the observed chance outcome. We reached a contradiction. The symmetric case where $v''$ is on the path from the root to $u''$ is analogous.\qedhere
\end{itemize}
\end{proof}

\subsection{Decomposition Algorithm}

    In this section, we provide pseudocode for the algorithm presented in \cref{sec:computation}. We will use the following conventions:
\begin{itemize}[leftmargin=*,nolistsep,itemsep=0mm]
  \item Given a player $i\in\{1,2\}$, we let $-i$ denote the opponent.
  \item We use the symbol $\sqcup$ to denote disjoint union.
  \item Given two infosets $I,I' \in \mathcal{I}_i$, we write $I \preceq I'$ if $\sigma(I') \succeq \sigma(I)$. We say that we iterate over a set $\mathcal{I} \subseteq \mathcal{I}_i$ \emph{in top-down order} if, given any two $I,I'\in\mathcal{I}$ such that $I\preceq I'$, $I$ appears before $I'$ in the iteration.
  \item We use the observation that for all $I\in\mathcal{I}_1$ and $\sigma_2\in\Sigma_2$, $I \rele \sigma_2$ if and only if $(I, a) \rele \sigma_2 \ \forall a\in A_I$. (A symmetric statement holds for $I \in \mathcal{I}_2$ and $\sigma_1\in\Sigma_1$.)
\end{itemize}

    \subsubsection{Two Useful Subroutines}
        We start by presenting two simple subroutines that capture fill-in step \circled{4} of \cref{fig:examples} or equivalently Step~\ref{pt:step4} of \cref{sec:computation}. The two subroutines are symmetric and have the role of filling rows and columns of the correlation plans.

        \begin{algorithm}[H]\small
          \caption{\normalfont$\textsc{FillOutRow}((\sigma_1,\sigma_2),I_1, \cS, \cD)$\hspace*{-.8cm}}
            \label{algo:fill out row}\DontPrintSemicolon
            \SetKwInOut{Pre}{Preconditions}
            \Pre{$(\sigma_1,\sigma_2)\in\Sigma_1\rele\Sigma_2, I_1\in \infos{1}, \sigma(I_1) = \sigma_1, (\sigma_1,\sigma_2)\in\cS$}
            \For{$I_2$ such that $\sigma(I_2) = \sigma_2$ and $\sigma_1\rele I_2$}{
                \uIf{$I_1 \conn I_2$}{
                    \For{$\sigma'_2 \in \{ (I_2, a)  : a \in A_{I_2}\}$}{
                        \Comment{Fill $(\sigma_1,\sigma'_2)$ by summing up all entries $\{v[(I_1, a'),\sigma'_2] : a' \in A_{I_1}\}$ in accordance with the \vsfc{}s}
                        $\cS \gets \cS \sqcup \{(\sigma_1,\sigma'_2)\}$\;
                        $\cD \gets \cD \ext^h \{1\}$ where $h : \vec{v} \mapsto \sum_{a' \in A_{I_1}} v[(I_1,a'),\sigma'_2]$\;
                    }
                }\Else{
                    \Comment{Fill all $\{v[\sigma_1,(I_2,a)] : a\in A_{I_2}\}$ by splitting $v[\sigma_1,\sigma_2]$ accordance with the \vsfc{}s}
                    $\cS \gets \cS \sqcup \{(\sigma_1,(I_2,a)) : a \in A_{I_2}\}$\;
                    $\cD \gets \cD \ext^h \Delta^{|A_{I_2}|}$ where $h : \vec{v} \mapsto v[\sigma_1,\sigma_2]$\;
                }
                \For{$\sigma'_2 \in \{ (I_2, a): a \in A_{I_2}\}$}{
                    $\textsc{FillOutRow}((\sigma_1, \sigma'_2), I_1)$\;\label{line:filloutrow recursive}
                }
            }
            \textbf{return} $(\cS,\cD)$\;
        \end{algorithm}
        \begin{algorithm}[H]\small
          \caption{\normalfont$\textsc{FillOutColumn}((\sigma_1,\sigma_2),I_2, \cS, \cD)$\hspace*{-.8cm}}
            \label{algo:fill out col}\DontPrintSemicolon
            \SetKwInOut{Pre}{Preconditions}
            \Pre{$(\sigma_1,\sigma_2)\in\Sigma_1\rele\Sigma_2, I_2\in \infos{2}, \sigma(I_2) = \sigma_2, (\sigma_1,\sigma_2)\in\cS$}
            \For{$I_1$ such that $\sigma(I_1) = \sigma_1$ and $\sigma_2 \rele I_1$}{
                \uIf{$I_1 \conn I_2$}{
                    \For{$\sigma' \in \{ (I_1, a)  : a \in A_{I_1}\}$}{
                        \Comment{Fill $(\sigma'_1,\sigma_2)$ by summing up all entries $\{v[\sigma'_1,(I_2,a')] : a' \in A_{I_2}\}$ in accordance with the \vsfc{}s}
                        $\cS \gets \cS \sqcup \{(\sigma'_1,\sigma_2)\}$\;
                        $\cD \gets \cD \ext^h \{1\}$ where $h : \vec{v} \mapsto \sum_{a' \in A_{I_2}} v[\sigma'_1,(I_2,a')]$\;
                    }
                }\Else{
                    \Comment{Fill all $\{v[(I_1,a),\sigma_2] : a\in A_{I_1}\}$ by splitting $v[\sigma_1,\sigma_2]$ accordance with the \vsfc{}s}
                    $\cS \gets \cS \sqcup \{((I_1,a), \sigma_2) : a \in A_{I_1}\}$\;
                    $\cD \gets \cD \ext^h \Delta^{|A_{I_1}|}$ where $h : \vec{v} \mapsto v[\sigma_1,\sigma_2]$\;
                }
                \For{$\sigma' \in \{ (I_1, a): a \in A_{I_1}\}$}{
                    $\textsc{FillOutColumn}((\sigma'_1, \sigma_2), I_2)$\;
                }
            }
            \textbf{return} $(\cS,\cD)$\;
        \end{algorithm}

The following inductive contract will be important for the full algortihm.

\begin{lemma}[Inductive contract for \textsc{FillOutRow}]
    Suppose that the following preconditions hold when $\textsc{FillOutRow}((\sigma_1,\sigma_2),I_1,\cS, \cD))$ is called:
    \begin{enumerate}[leftmargin=1.5cm,nolistsep,itemsep=1mm,label=(Pre\arabic*)]
      \item[(Pre1)] $(\sigma_1,\sigma_2)\in\Sigma_1\rele\Sigma_2$;
      \item[(Pre2)] $I_1 \in \infos{1}$ is such that $\sigma(I_1) = \sigma$;
      \item[(Pre3)] $\cS$ contains only relevant sequence pairs and $\cD$ consists of vectors indexed by exactly the indices in $\cS$;
      \item[(Pre4)] $(\sigma_1,\sigma_2) \in \cS$, but $(\sigma_1,\sigma'_2) \notin \cS$ for all $\sigma'_2 \succ \sigma_2$;
      \item[(Pre5)] For all $a \in I_1$ and $\sigma'_2 \succeq \sigma_2$ such that $I_1 \rele \sigma'_2$, $((I_1,a), \sigma'_2) \in \cS$;
      \item[(Pre6)] If $I_1 \rele \sigma_2$, all $\vec{v}\in\cD$ satisfy the \vsfc{} $v[\sigma_1,\sigma_2] = \sum_{a \in I_1} v[(I_1,a), \sigma_2]$;
      \item[(Pre7)] All $\vec{v}\in\cD$ satisfy the \vsfc{}s
        \[
            v[(I,a), \sigma(I_2)] = \sum_{a'\in A_{I_2}} v[(I,a), (I_2,a')]\quad\forall a\in I_1, \text{ and } I_2 \in \infos{2}: I_1 \rele I_2, \sigma(I_2) \succeq \sigma_2.
        \]
    \end{enumerate}
    Then, the sets $(\cS', \cD')$ returned by the call are such that
    \begin{enumerate}[leftmargin=1.5cm,nolistsep,itemsep=1mm,label=(Post\arabic*)]
      \item[(Post1)] $\cS'$ contains only relevant sequence pairs and $\cD'$ consists of vectors indexed by exactly the indices in $\cS'$;
      \item[(Post2)] $\cS' = \cS \sqcup \{(\sigma_1,\sigma'_2): \sigma'_2 \succ \sigma_2, \sigma\rele\sigma'_2\}$;
      \item[(Post3)] All $\vec{v}\in\cD'$ satisfy the \vsfc{}s
        \[
            v[\sigma_1, \sigma(I_2)] = \sum_{a'\in A_{I_2}} v[\sigma_1, (I_2,a')]\quad\forall I_2 \in \infos{2}: \sigma \rele I_2, \sigma(I_2) \succeq \sigma_2
        \]
        and all \vsfc{}s
        \[
            v[\sigma_1, \sigma'_2] = \sum_{a\in A_{I_1}} v[(I,a), \sigma'_2]\quad\forall \sigma'_2 \in \Sigma_2: \sigma'_2 \rele I_1, \sigma'_2 \succeq \sigma_2.
        \]
    \end{enumerate}
\end{lemma}
\begin{proof}
    By induction.
    \begin{itemize}[nolistsep,itemsep=1mm,leftmargin=5mm]
        \item \textbf{Base case.} The base case corresponds to $\sigma_2 \in \Sigma_2$ such that no information set $I_2 \in \infos{2}: \sigma(I_2) = \sigma_2 \land \sigma_1 \rele I_2$ exists. In that case, \cref{algo:fill out row} returns immediately, so (Post1) holds trivially from (Pre3).
            Since no $I_2$ such that $\sigma(I_2) = \sigma_2\land \sigma_1 \rele I_2$ exists, no $\sigma'_2 \succ \sigma_2$ such that $\sigma_1 \rele \sigma'_2$ exists, so (Post2) holds as well. The first set of constraints of (Post3) is empty, and the second set reduces to (Pre6).
        \item \textbf{Inductive step.} Suppose that the inductive hypothesis holds when $\sigma'_2 \succ \sigma_2$. We will show that it holds when $\sigma'_2 = \sigma_2$ as well. In order to use the inductive hypothesis, we first need to check that the preconditions are preserved at the time of the recursive call on \cref{line:filloutrow recursive}. (Pre1) holds since $\sigma_1 \rele I_2$. (Pre2) holds trivially since $\sigma$ does not chance. (Pre3) holds since we are updating $\cS$ and $\cD$ in tandem on lines 4, 5 and 7, 8. (Pre4) holds since by the time of the recursive call we have only filled in entries $(\sigma_1,\sigma'_2)$ where $\sigma'_2$ is an immediate successor of $\sigma_2$. (Pre5) at Line 10 holds trivially, since it refers to a subset of the entries for which the condition held at the beginning of the call. (Pre6) holds because $I_1 \rele \sigma'_2 \iff I_1\conn I_2$. Hence, if $I_1 \rele \sigma'_2$ then Lines 4 and 5 must have run. (Pre7) at Line 10 holds trivially, since it refers to a subset of the constraints for which the condition held at the beginning of the call. Using the inductive hypothesis, (Post1), (Post2), and the second set of constraints in (Post3) follow immediately. The only constraints that are left to be verified are
            \begin{align}\label{eq:to prove}
                v[\sigma_1, \sigma_2] = \sum_{a'\in A_{I_2}} v[\sigma_1, (I_2,a')]\quad\forall I_2 \in \infos{2}: \sigma \rele I_2, \sigma(I_2) = \sigma_2.
            \end{align}
            That constraint is guaranteed by Lines 7 and 8 for all $I_2 \not\conn I_1$. So, we need to verify that it holds for all those $I_2$ such that $\sigma(I_2)=\sigma_2, \sigma\rele I_2$ and $I_1\conn I_2$. Let $I_2$ be one such information set. Then, from Lines 4 and 5 we have that
            \[
                v[\sigma_1, (I_2, a)] = \sum_{a'\in A_{I_1}} v[(I,a'), (I_2,a)]\quad \forall a\in A_{I_2}.
            \]
            Summing the above equations across all $a \in A_{I_2}$ and using (Pre7) yields
            \begin{align*}
                \sum_{a\in A_{I_2}} v[\sigma_1, (I_2,a)] &= \sum_{a\in A_{I_2}}\sum_{a'\in A_{I_1}} v[(I,a'), (I_2,a)]\\
                    &= \sum_{a'\in A_{I_1}} \sum_{a\in A_{I_2}} v[(I,a'), (I_2,a)]\\
                    &= \sum_{a'\in A_{I_1}} v[(I,a'), \sigma(I_2)]\\
                    &= \sum_{a'\in A_{I_1}} v[(I,a'), \sigma_2],
            \end{align*}
            where we used the hypothesis that $\sigma(I_2) = \sigma_2$ in the last equality. Finally, since $I_1 \conn I_2$ and $\sigma(I_2) = \sigma_2$, it must be $I_1 \rele \sigma_2$ and so, using (Pre6), we obtain that
            \[
                \sum_{a\in A_{I_2}} v[\sigma_1, (I_2,a)] = v[\sigma_1, \sigma_2],
            \]
            completing the proof of~\cref{eq:to prove}. So, (Post3) holds as well and the proof of the inductive step is complete.
            \qedhere
    \end{itemize}
\end{proof}

The inductive contract for \textsc{FillOutColumn} is symmetric and we omit it.

\subsubsection{The Full Algorithm}

        \begin{algorithm}[H]\small
          \caption{\normalfont$\textsc{Decompose}((\sigma_1,\sigma_2),\cS, \cD)$\hspace*{-.8cm}}
            \label{algo:decompose}\DontPrintSemicolon
            \SetKwInOut{Pre}{Preconditions}
            \Pre{$(\sigma_1,\sigma_2)\in\Sigma_1\rele\Sigma_2, (\sigma_1,\sigma_2)\in\cS$}
            \BlankLine
            $B \gets \emptyset$\;
            \For{all $i\in \{1,2\}, I \in \infos{i}, \sigma(I) = \sigma_i, \sigma_{-i}\rele I$}{
                \uIf{the $\sigma_{-i}$-rank of $I$ is $0$}{
                    $B \gets B \sqcup I$\;
                }
            }
            \For{$(I_1,I_2) \in \infos{1}\times\infos{2}$ such that $\sigma(I_1) = \sigma_1, \sigma(I_2) = \sigma_2, I_1\conn I_2$}{
                \uIf{the $\sigma_2$-rank of $I_1$ is $\ge$ the $\sigma_1$-rank of $I_2$}{
                    $B \gets B \sqcup I_1$\;
                }\Else{
                    $B \gets B \sqcup I_2$\;
                }
            }
            \For{$I \in B$}{
                \uIf{$I \in \infos{1}$}{
                    \Comment{Fill all $\{v[(I,a),\sigma_2] : a\in A_{I}\}$ by splitting $v[\sigma_1,\sigma_2]$ accordance with the \vsfc{}s}
                    $\cS \gets \cS \sqcup \{((I,a), \sigma_2) : a \in A_{I}\}$\;
                    $\cD \gets \cD \ext^h \Delta^{|A_{I}|}$ where $h : \vec{v} \mapsto v[\sigma_1,\sigma_2]$\;
                    \BlankLine
                    \Comment{Recursive call}
                    \For{$\sigma'_1 \in \{(I,a) : a \in A_I\}$}{
                        $\textsc{Decompose}((\sigma'_1,\sigma_2), \cS, \cD)$\;
                    }
                    \BlankLine
                    \Comment{Fill a portion of the row for $\sigma_1$}
                    \For{$I_2 \in \infos{2} : \sigma_1 \rele I_2, \sigma(I_2) = \sigma_2$}{
                        \For{$\sigma'_2 \in \{(I_2, a'): a' \in A_{I_2}\}$}{
                            \Comment{Fill $(\sigma_1,\sigma'_2)$ by summing up all entries $\{v[(I, a'),\sigma'_2] : a' \in A_{I}\}$ in accordance with the \vsfc{}s}
                            $\cS \gets \cS \sqcup \{(\sigma_1,\sigma'_2)\}$\;
                            $\cD \gets \cD \ext^h \{1\}$ where $h : \vec{v} \mapsto \sum_{a' \in A_{I}} v[(I,a'),\sigma'_2]$\;
                            \BlankLine
                            $\textsc{FillOutRow}((\sigma_1, \sigma'_2), I)$\;
                        }
                    }
                }\Else{
                    \Comment{Fill all $\{v[\sigma_1, (I,a)] : a\in A_{I}\}$ by splitting $v[\sigma_1,\sigma_2]$ accordance with the \vsfc{}s}
                    $\cS \gets \cS \sqcup \{(\sigma_1, (I,a)) : a \in A_{I}\}$\;
                    $\cD \gets \cD \ext^h \Delta^{|A_{I}|}$ where $h : \vec{v} \mapsto v[\sigma_1,\sigma_2]$\;
                    \BlankLine
                    \Comment{Recursive call}
                    \For{$\sigma'_2 \in \{(I,a) : a \in A_I\}$}{
                        $\textsc{Decompose}((\sigma_1,\sigma'_2), \cS, \cD)$\;
                    }
                    \BlankLine
                    \Comment{Fill a portion of the column for $\sigma_2$}
                    \For{$I_1 \in \infos{1} : \sigma_2 \rele I_1, \sigma(I_1) = \sigma_1$}{
                        \For{$\sigma'_1 \in \{(I_1, a'): a' \in A_{I_1}\}$}{
                            \Comment{Fill $(\sigma'_1,\sigma_2)$ by summing up all entries $\{v[\sigma'_1, (I, a')] : a' \in A_{I}\}$ in accordance with the \vsfc{}s}
                            $\cS \gets \cS \sqcup \{(\sigma'_1,\sigma_2)\}$\;
                            $\cD \gets \cD \ext^h \{1\}$ where $h : \vec{v} \mapsto \sum_{a' \in A_{I}} v[\sigma'_1, (I,a')]$\;
                            \BlankLine
                            $\textsc{FillOutColumn}((\sigma'_1, \sigma_2), I)$\;
                        }
                    }
                }
            }
            \textbf{return} $(\cS,\cD)$\;
        \end{algorithm}

\begin{lemma}[Inductive contract for \textsc{Decompose}]\label{lem:contract}
  Assume that at the beginning of each call to $\textsc{Decompose}((\sigma_1,\sigma_2), \mathcal{S}, \mathcal{D})$ the following conditions hold
\begin{enumerate}[leftmargin=1.5cm,nolistsep,itemsep=1mm,label=(Pre\arabic*)]
  \item $\mathcal{S}$ contains only relevant sequence pairs and $\mathcal{D}$ consists of vectors indexed by exactly the indices in $\mathcal{S}$.
  \item $\mathcal{S}$ does not contain any relevant sequence pairs which are descendants of $(\sigma_1,\sigma_2)$, with the only exception of $(\sigma_1,\sigma_2)$ itself. In formulas, \[\mathcal{S} \cap \{(\sigma'_1,\sigma'_2) \in \Sigma_1\times\Sigma_2 : \sigma'_1 \succeq \sigma_1, \sigma'_2 \succeq \sigma_2\} = \{(\sigma_1, \sigma_2)\}.\]
\end{enumerate}

  Then, at the end of the call, the returned set $(\mathcal{S}', \mathcal{D}')$ are such that
\begin{enumerate}[leftmargin=1.5cm,nolistsep,itemsep=1mm,label=(Post\arabic*)]
\item $\mathcal{S}'$ contains only relevant sequence pairs and $\mathcal{D}'$ consists of vectors $\vec{v}$ indexed by exactly the indices in $\mathcal{S}'$.
\item The call has filled in exactly all relevant sequence pair indices that are descendants of $(\sigma_1,\sigma_2)$ (except for $(\sigma_1,\sigma_2)$ itself, which was already filled in). In formulas, \[\mathcal{S}' = \mathcal{S} \sqcup \{(\sigma'_1,\sigma'_2) \in \Sigma_1\times\Sigma_2 : \sigma'_1 \succeq \sigma_1, \sigma'_2 \succeq \sigma_2, (\sigma'_1,\sigma'_2) \neq (\sigma_1,\sigma_2), \sigma'_1 \rele \sigma'_2\}.\]
\item $\mathcal{D}'$ satisfies the subset of \vsfc{}s
\end{enumerate}
    \begin{equation*}
  \begin{array}{ll}
    \ \!\displaystyle\sum_{a \in A_I}~\!v[(I, a),\hspace{.1mm} \sigma'_2] = v[\sigma(I),\hspace{.4mm} \sigma'_2] & \forall \sigma'_2 \succeq \sigma_2, I \in \mathcal{I}_1 \ \hspace{.0mm} \text{ s.t. } \sigma'_2 \rele I, \sigma(I) \succeq \sigma_1\\[4mm]
    \ \!\displaystyle\sum_{a \in A_J} v[\sigma'_1, (J, a)] = v[\sigma'_1, \sigma(J)] & \forall \sigma'_1\succeq\sigma_1, J \in \mathcal{I}_2 \text{ s.t. } \sigma'_1 \rele J, \sigma(J) \succeq \sigma_2.
  \end{array}
  \end{equation*}
\end{lemma}
\begin{proof}
    By induction.
    \begin{itemize}[nolistsep,itemsep=1mm,leftmargin=5mm]
      \item \textbf{Base case.} The base case is any $(\sigma_1,\sigma_2)$ such that there is no $\sigma'_1 \succeq \sigma_1, \sigma'_2 \succeq \sigma_2$, $\sigma'_1 \rele \sigma'_2$. In that case, the set $B$ is empty, so the algorithm terminates immediately without modifying the sets $\cS$ and $\cD$. Consequently, (Post1) and (Post2) hold trivially from (Pre1) and (Pre2). (Post3) reduces to an empty set of constraints, so (Post3) holds as well.
      \item \textbf{Inductive step.} In order to use the inductive hypothesis, we will need to prove that the preconditions for \textsc{Decompose} hold on Lines 15 and 25. We will focus on Line~15 ($I \in \infos{1}$), as the analysis for the other case ($I\in\infos{2}$) is symmetric. (Pre1) clearly holds, since we always update $\cS$ and $\cD$ in tandem.
          %
          %
          Since all iterations of the \textbf{for} loop on Line~10 touch different information sets, at the time of the recursive call on Line~15, and given (Post2) for all previous recursive calls, the only relevant sequence pairs $(\sigma''_1, \sigma''_2)$ such that $\sigma''_1 \succeq \sigma'_1, \sigma''_2 \succeq \sigma_2$ that have been filled are the ones on Lines~12 and~13. So, (Pre2) holds.

          We now check that the preconditions for \textsc{FillOutRow} hold at Line~20. (Pre1), (Pre2), (Pre3), and (Pre4) are trivial. (Pre5) and (Pre7) are guaranteed by (Post2) and (Post3) of \textsc{Decompose} applied to Line~15. (Pre6) holds because of Lines 18 and 19.

          Using the inductive contracts of \textsc{FillOutRow}, \textsc{FillOutColumn} and \textsc{Decompose} for the recursive calls, we now show that all postconditions hold at the end of the call. (Post1) is trivial since we always update $\cS$ and $\cD$ together. (Post2) holds by keeping track of what entries are filled in Lines~12, 13, 18, 19, 22, 23, 28, 29, as well as those filled in the calls to \textsc{FillOutRow}, \textsc{FillOutColumn} and \textsc{Decompose}, as regulated by postcondition (Post2) in the inductive contracts of the functions. In order to verify (Post3), we need to verify that the constraints that are not already guaranteed by the recursive calls hold. In particular, we need to verify that
  \[\begin{array}{ll}
    \circled{A}\ \!\displaystyle\sum_{a \in A_I}~\!v[(I, a),\hspace{.1mm} \sigma_2] = v[\sigma_1,\hspace{.4mm} \sigma_2] & \forall I \in \mathcal{I}_1 \ \hspace{.0mm} \text{ s.t. } \sigma_2 \rele I, \sigma(I) = \sigma_1, I \notin B\\[4mm]
    \circled{B}\ \!\displaystyle\sum_{a \in A_J} v[\sigma_1, (J, a)] = v[\sigma_1, \sigma_2] & \forall J \in \mathcal{I}_2 \text{ s.t. } \sigma_1 \rele J, \sigma(J) = \sigma_2, J \notin B.
  \end{array}\]
  We will show that constraints \circled{A} hold; the proof for \circled{B} is symmetric. Using \cref{lem:ranks} together with the definition of $B$ (Lines 1-9), any information set $I \in \infos{i}: \sigma(I) = \sigma_i, \sigma_{-i}\rele I$ that is not in $B$ must have $\sigma_{-i}$-rank \emph{exactly} $1$. Let $I \in \mathcal{I}_1$ be such that $\sigma_2 \rele I, \sigma(I) = \sigma_1, I \notin B$, as required in \circled{A}. Since the $\sigma_2$-rank of $I$ is $1$, let $J$ be the only information set in $\infos{2}$ such that $I \conn J, \sigma(J) = \sigma_2$. Note that $J \in B$.
The entries $v[(I,a), \sigma_2] : a\in A_I$ were filled in Lines 28 and 29 when the \textbf{for} loop picked up $J \in B$. So, in particular,
    \[
        v[(I,a), \sigma_2] = \sum_{a'\in A_{J}} v[(I,a), (J,a')] \quad \forall a\in A_I.
    \]
    Summing the above equations across $a\in A_I$, we obtain
    \begin{align*}
        \sum_{a\in A_I} v[(I,a), \sigma_2] &= \sum_{a \in A_I} \sum_{a'\in A_{J}} v[(I,a), (J,a')] \\
            &= \sum_{a'\in A_{J}} \sum_{a \in A_I}  v[(I,a), (J,a')]\\
            &= \sum_{a'\in A_{J}} v[\sigma_1, (J,a')]\\
            &= v[\sigma_1, \sigma_2],
    \end{align*}
    where the last equation follows from the way the entries $v[\sigma_1, (J,a')] : a' \in A_{J}$ were filled in (Lines 22 and 23). This shows that the set of constraints \circled{A} hold.
\qedhere
    \end{itemize}
\end{proof}

\thmdecomposition*
\begin{proof}
    The correctness of the algorithm follow from (Post3) in the inductive contract. Every time the set of partially-filled-in vectors $\cD$ gets extended, it is extended with either the singleton set $\{1\}$ or a simplex. In either case the nonnegative affine functions $h$ used are linear. So, the decomposition structure is as in the statement. Finally, since the overhead of each call (on top of the recursive calls) is linear in the number of relevant sequence pairs $(\sigma,\tau) \in \Sigma_1\rele\Sigma_2$ that are filled, and each relevant sequence pair is filled only once, the complexity of the algorithm is linear in the number of relevant sequence pairs.
\end{proof} 

%% file: text/appendix_proofs_integrality.tex
\section{Relationship Between $\cV$ and $\Xi$}\label{app:xi equal vonS}

\subsection{Preliminaries: Definition of the Polytope of Correlation Plans}

Let $\Pi_i(\sigma)$ denote the subset of reduced-normal-form plans $\Pi_i$ for Player~$i$ prescribe all actions of Player~$i$ on the path from the root of the game down to the information set-action pair $\sigma$ (if $\sigma=\empty$, assign $\Pi_i(\emptyseq) = \Pi_i$).
The transformation from a correlated distribution $\mu$ to its correlation plan representation is achieved using a linear function
\[
    f : \Delta^{|\Pi_1 \times \Pi_2|} \to \bbR_{\ge 0}^{|\Sigma_1\rele\Sigma_2|}.
\]
Specifically, $f$ takes a generic distribution $\mu$ over $\Pi_1\times\Pi_2$ and maps to the vector $\vec{\xi} = f(\mu)$, called a \emph{correlation plan}, whose components are
\begin{equation}\label{eq:def f}
    \xi[\sigma_1, \sigma_2] \defeq \sum_{\pi_1 \in \Pi_1(\sigma_1)} \sum_{\pi_2\in\Pi_2(\sigma_2)} \mu(\pi_1,\pi_2)\qquad \forall (\sigma_1,\sigma_2) \in \Sigma_1 \rele \Sigma_2.
\end{equation}

The set of all valid correlation plans, $\Xi$, is defined as the image $\img f$ of $f$ as the distribution $\mu$ takes any possible value in $\Delta^{|\Pi_1\times\Pi_2|}$.

\begin{remark}\label{rem:V in 01}
    Since $f$ sums up distinct entries from the distribution $\mu$, all entries in $\vec{\xi} = f(\mu)$ are in the range $[0,1]$.
\end{remark}

\subsection{Proofs}

\begin{lemma}\label{lem:xi co}
    Let $\vec{1}_{(\pi_1,\pi_2)} \in \Delta^{|\Pi_1\times\Pi_2|}$
    denote the distribution over $\Pi_1\times\Pi_2$ that assigns
    mass $1$ to the pair $(\pi_1, \pi_2)$, and mass $0$ to any
    other pair of reduced-normal-form plans. Then,
    \[
        \Xi = \co\{f(\vec{1}_{(\pi_1, \pi_2)}) : \pi_1 \in \Pi_1, \pi_2 \in \Pi_2\}.
    \]
\end{lemma}
\begin{proof}
    The ``deterministic'' distributions $\vec{1}_{(\pi_1, \pi_2)}$ are the vertices of $\Delta^{|\Pi_1\times\Pi_2|}$, so, in particular,
    \[
        \Delta^{|\Pi_1\times\Pi_2|} = \co\{\vec{1}_{(\pi_1,\pi_2)} : \pi_1 \in \Pi_1, \pi_2 \in \Pi_2\}.
    \]
    Since by definition $\Xi = \img f$, and $f$ is a linear function, the
    images (under $f$) of the $\vec{1}_{(\pi_1,\pi_2)}$ are a convex basis
    for $\Xi$, which is exactly the statement.
\end{proof}

\begin{lemma}\label{lem:v zero}
    Let $\vec{v} \in \cV$. For all $\sigma_1 \in \Sigma_1$ such that $v[\sigma_1, \emptyseq] = 0$, $v[\sigma_1, \sigma_2] = 0$ for all $\sigma_2 \rele \sigma_1$. Similarly, for all $\sigma_2 \in \Sigma_2$ such that $v[\emptyseq, \sigma_2] = 0$, $v[\sigma_1, \sigma_2] = 0$ for all $\sigma_1 \rele \sigma_2$.
\end{lemma}
\begin{proof}
  We prove the theorem by induction on the \emph{depth} of the sequences $\sigma_1$ and $\sigma_2$. The depth $\dep(\sigma)$ of a generic sequence $\sigma = (I,a) \in \Sigma_i$ of Player~$i$ is defined as the number of actions that Player $i$ plays on the path from the root of the tree down to action $a$ at information set $I$ included. Conventionally, we let the depth of the empty sequence be $0$.

    Take $\sigma_1 \in \Sigma_1$ such that $v[\sigma_1, \emptyseq] = 0$. For $\sigma_2$ of depth $0$ (that is, $\sigma_2 = \emptyseq$), clearly $v[\sigma_1, \sigma_2] = 0$. For the inductive step, suppose that $v[\sigma_1, \sigma_2] = 0$ for all $\sigma_2 \in \Sigma_2, \sigma_1 \rele \sigma_2$ such that $\dep(\sigma_2) \le d_2$. We will show that $v[\sigma_2, \sigma_2] = 0$ for $\dep(\sigma_2) \le d_2 + 1$. Indeed, let $(I, a') = \sigma_2 \rele \sigma_1$ of depth $d_2 + 1$. Since $\vec{v} \in \cV$, in particular the \vsfc{} $\sum_{a \in A_I} v[\sigma_1, (I, a)] = v[\sigma_1, \sigma(I)]$ must hold. The depth of $\sigma(I)$ is $d_2$, so by the inductive hypothesis, it must be $v[\sigma_1, \sigma(I)] = 0$, and therefore $\sum_{a \in A_I} v[\sigma_1, (I, a)] = 0$. But all entries of $\vec{v}$ are nonnegative, so it must be $v[\sigma_1, (I,a)] = 0$ for all $a \in A_I$, and in particular for $(I,a') = \sigma_2$. This completes the proof by induction.

    The proof for the second part is analogous.
\end{proof}

\begin{lemma}\label{lem:v product}
    Let $\vec{v} \in \cV$ have integer $\{0,1\}$ coordinates. Then, for all $(\sigma_1, \sigma_2) \in \Sigma_1 \rele \Sigma_2$, it holds that
    \[
        v[\sigma_1,\sigma_2] = v[\sigma_1, \emptyseq] \cdot v[\emptyseq, \sigma_2].
    \]
\end{lemma}
\begin{proof}
    We prove the theorem by induction on the depth of the sequences, similarly to \cref{lem:v zero}.

    The base case for the induction proof corresponds to the case where $\sigma_1$ and $\sigma_2$ both have depth $0$, that is, $\sigma_1 = \sigma_2 = \emptyseq$. In that case, the theorem is clearly true, because $v[\emptyseq,\emptyseq]= 1$ as part of the \vsfc{}s~\eqref{eq:vonS constraints}.

    Now, suppose that the statement holds as long as $\dep(\sigma_1),\dep(\sigma_2) \le d$. We will show that the statement will hold for any $(\sigma_1,\sigma_2) \in \Sigma_1 \rele \Sigma_2$ such that $\dep(\sigma_1), \dep(\sigma_2) \le d+1$. Indeed, consider $(\sigma_1,\sigma_2) \in \Sigma_1 \rele \Sigma_2$ such that $\dep(\sigma_1), \dep(\sigma_2) \le d+1$. If any of the sequences is the empty sequence, the statements holds trivially, so assume that neither is the empty sequence and in particular $\sigma_1 = (I,a)$, $\sigma_2 = (J,b)$. If $v[\sigma_1, \emptyseq] = 0$, then from \cref{lem:v zero} $v[\sigma_1, \sigma_2] = 0$ and the statement holds. Similarly, if $v[\emptyseq, \sigma_2] = 0$, then $v[\sigma_1, \sigma_2] = 0$, and the statement holds. Hence, the only remaining case given the integrality assumption on the coordinates of $\vec{v}$ is $v[\sigma_1,\emptyseq] = v[\emptyseq,\sigma_2] = 1$.

From the \vsfc{}s, $v[\sigma(I), \emptyseq] = \sum_{a'\in A_I} v[(I, a'), \emptyseq] = 1 + \sum_{a' \in A_I, a' \neq a} v[(I, a'), \emptyseq] \ge 1$. Hence, because all entries of $\vec{v}$ are in $\{0,1\}$, it must be $v[\sigma(I), \emptyseq] = 1$ and $v[(I,a'),\emptyseq] = 0$ for all $a' \in A_I, a' \neq a$. With a similar argument we conclude that $v[\emptyseq, \sigma(J)] = 1$ and $v[\emptyseq, (J,b')] = 0$ for all $b' \in A_J, b \neq b'$. Using the inductive hypothesis, $v[\sigma(I), \sigma(J)] = v[\sigma(I), \emptyseq]\cdot v[\emptyseq, \sigma(J)] = 1$.

Now, using the \vsfc{}s together with the equality $v[\sigma(I), \sigma(J)]  = 1$ we just proved, we conclude that
\begin{equation}\label{eq:lem5 step1}
    \sum_{a' \in A_I}\sum_{b' \in A_J} v[(I,a'),(J,b')] = 1.
\end{equation}
On the other hand, since $v[(I,a'), \emptyseq] = 0$ for all $a' \in A_I, a' \neq a$ and $v[\emptyseq, (J,b')] = 0$ for all $b' \in A_J, b' \neq b$, from \cref{lem:v zero} we have that
\begin{equation}\label{eq:lem5 step2}
    a' \neq a \lor b' \neq b \implies v[(I,a'), (J,b')] = 0.
\end{equation}
From \eqref{eq:lem5 step2} and \eqref{eq:lem5 step1}, we conclude that $v[(I,a), (J,b)] = v[\sigma_1, \sigma_2] = 1 = v[\sigma_1,\emptyseq] \cdot v[\emptyseq, \sigma_2]$, as we wanted to show.
\end{proof}

\thmvinxi*
\begin{proof} We prove the two implications separately.
    \begin{itemize}[nolistsep,itemsep=1mm,leftmargin=8mm]
        \item[($\Rightarrow$)] We start by proving that if $\Xi = \cV$, then all vertices of $\cV$ have integer $\{0,1\}$ coordinates.
            Since $\cV = \Xi$ by hypothesis, from \ref{lem:xi co} we
            can write
            \[
                \cV = \co\{f(\vec{1}_{(\pi_1, \pi_2)}) : \pi_1 \in \Pi_1, \pi_2 \in \Pi_2\}.
            \]
            So, to prove this direction it is enough to show that
            $f(\vec{1}_{(\pi_1,\pi_2)})$ has integer $\{0,1\}$ coordinates
            for all $(\pi_1,\pi_2)\in\Pi_1\times\Pi_2$. To see that, we use
            the definition~\eqref{eq:def f}: each entry in
            $f(\vec{1}_{(\pi_1,\pi_2)})$ is the sum of distinct entries of
            $\vec{1}_{(\pi_1,\pi_2)}$. Given that by definition
            $\vec{1}_{(\pi_1,\pi_2)}$ has exactly one entry with value $1$
            and $|\Pi_1\times\Pi_2|-1$ entries with value $0$, we conclude
            that all coordinates of $f(\vec{1}_{(\pi_1,\pi_2)})$ are in
            $\{0,1\}$.

        \item[($\Leftarrow$)] We now show that if all vertices of $\cV$
            have integer $\{0,1\}$ coordinates, then $\cV \subseteq \Xi$.
            This is enough, since the reverse inclusion,
            $\cV \supseteq \Xi$, is trivial and already
            known~\cite{Stengel08:Extensive}. Let
            $\{\vec{v}_1, \dots, \vec{v}_n\}$ be the vertices of $\cV$. To
            conclude that $\cV \subseteq \Xi$, we will prove that
            $\vec{v}_i \in \Xi$ for all $i=1,\dots, n$. This will be
            sufficient since both $\cV$ and $\Xi$ are convex.

            \vspace{1mm}
            Let $\vec{v} \in \{\vec{v}_1, \dots, \vec{v}_n\}$ be any vertex
            of $\cV$. By hypothesis, $v[\sigma_1, \sigma_2] \in \{0,1\}$
            for all $(\sigma_1,\sigma_2) \in \Sigma_1 \rele \Sigma_2$.
            Because $\vec{v}$ satisfies the \vsfc{}s and furthermore
            $\vec{v}$ has $\{0,1\}$ entries by hypothesis, the two vectors
            $\vec{q}_1,\vec{q}_2$ defined according to
            $\vec{q}_1[\sigma_1] = v[\sigma_1, \emptyseq]$
            ($\sigma_1 \in \Sigma_1)$ and
            $\vec{q}_2[\sigma_2] = v[\emptyseq, \sigma_2]$
            ($\sigma_2 \in \Sigma_2)$ are \emph{pure} sequence-form
            strategies. Now, let $\pi^*_1$ and $\pi^*_2$ be
            the reduced-normal form plans corresponding to $\vec{q}_1$ and
            $\vec{q}_2$, respectively. We will show that
            $\vec{v} = f(\vec{1}_{(\pi^*_1, \pi^*_2)})$, which will
            immediately imply that $\vec{v} \in \Xi$ using
            \cref{lem:xi co}.

            Since $\vec{1}_{(\pi^*_1,\pi^*_2)}$ has exactly one positive entry with value $1$ in the position corresponding to $(\pi_1^*, \pi_2^*)$, by definition of the linear map $f$, for any $(\sigma_1, \sigma_2) \in \Sigma_1 \rele\Sigma_2$,
            \begin{equation}
                f(\vec{1}_{(\pi^*_1,\pi^*_2)})[\sigma_1, \sigma_2] = \bbone[\sigma_1 \in \Pi_1(\sigma_1)]\cdot \bbone[\sigma_2 \in \Pi_2(\sigma_2)].
            \end{equation}
            So, using the known properties of pure sequence-form strategies, we obtain
            \begin{align*}
                f(\vec{1}_{(\pi^*_1,\pi^*_2)})[\sigma_1, \sigma_2] = q_1[\sigma_1] \cdot q_2[\sigma_2] = v[\sigma_1, \emptyseq]\cdot v[\emptyseq, \sigma_2] = v[\sigma_1, \sigma_2],
            \end{align*}
            where the last equality follows from \cref{lem:v product}. Since the equality holds for any $(\sigma_1, \sigma_2) \in \Sigma_1 \rele\Sigma_2$, we have that $\vec{v} = f(\vec{1}_{(\pi^*_1,\pi^*_2)})$.
            \qedhere
    \end{itemize}
\end{proof}

\lemvertscext*
\begin{proof}
    Take any point $\vec{z} \in \cX \ext^h \cY$. By definition of scaled, extension, there exist $\vec{x}\in\cX, \vec{y}\in\cY$ such that $\vec{z} = (\vec{x}, h(\vec{x})\vec{y})$.
    Since $\{\vec{x}_1,\dots,\vec{x}_n\}$ are the vertices of $\cX$, $\vec{x}$ can be written as a convex combination $\vec{x} = \sum_{i=1}^n \lambda_i \vec{x}_i$ where $(\lambda_1,\dots,\lambda_n)\in\Delta^n$. Similarly, $\vec{y} = \sum_{i=1}^m \mu_i \vec{y}_i$ for some $(\mu_1,\dots,\mu_m)\in\Delta^m$. Hence, using the hypothesis that $h$ is affine, we can write
\begin{align*}
    \vec{z} &= (\vec{x}, h(\vec{x})\vec{y})\\
            &= \mleft(\sum_{i=1}^n \lambda_i\vec{x}_i, h\mleft(\sum_{i=1}^n \lambda_i\vec{x}_i\mright)\sum_{j=1}^m\mu_j\vec{y}_j\mright)\\
            &= \mleft(\sum_{i=1}^n \lambda_i\vec{x}_i, \mleft(\sum_{i=1}^n \lambda_i h(\vec{x}_i)\mright)\sum_{j=1}^m\mu_j\vec{y}_j\mright)\\
            &= \sum_{i=1}^n\sum_{j=1}^m \lambda_i\mu_j (\vec{x}_i, h(\vec{x}_i)\vec{y}_j).
\end{align*}
Since $\lambda_i\mu_j \ge 0$ for all $i\in\{1,\dots,n\},j\in\{1,\dots,m\}$ and $\sum_{i=1}^n\sum_{j=1}^m \lambda_i\mu_j = (\sum_{i=1}^n \lambda_i)(\sum_{j=1}^m \mu_j) = 1$, we conclude that $\vec{z} \in \co\{(\vec{x}_i,h(\vec{x}_i)\vec{y}_j) : i\in\{1,\dots,n\},j\in\{1,\dots,m\}\}$. On the other hand, $(\vec{x}_i, h(\vec{x_i})\vec{y}_j) \in \cX\ext^h\cY$, so
\[
    \cX \ext^h \cY = \co\large\{(\vec{x}_i,h(\vec{x}_i)\vec{y}_j) : i\in\{1,\dots,n\}, j\in\{1,\dots,m\}\large\}.
\]
Since the vertices of a (nonempty) polytope are a (nonempty) subset of any convex basis for the polytope, the vertices of $\cX \ext^h \cY$ must be a nonempty subset of $\{(\vec{x}_i,h(\vec{x}_i)\vec{y}_j) : i\in\{1,\dots,n\}, j\in\{1,\dots,m\}\}$, which is the statement.
\end{proof}

\thmintegrality*
\begin{proof}
    We prove the statement by induction over the scaled-extension-based decomposition
    \[
        \cV = \{1\} \ext^{h_1} \cX_1 \ext^{h_2} \cdots \ext^{h_n} \cX_n.
    \]
    In particular, we will show that for all $k = 0, \dots, n$, the coordinates of the vertices of the polytope
    \[
        \cV_k = \{1\} \ext^{h_1} \cdots \ext^{h_k} \cX_k
    \]
    constructed by considering only the first $k$ scaled extensions in the decomposition are all integer. Since $\cV \subseteq [0,1]^{|\Sigma_1\rele\Sigma_2|}$ (\cref{rem:V in 01}), this immediately implies that each coordinate is in $\{0,1\}$.
    \begin{itemize}[nolistsep,itemsep=1mm,leftmargin=5mm]
      \item \textbf{Base case:} $k=0$. In this case, $\cV_0 = \{1\}$. The only vertex is $\{1\}$, which is integer. So, base case trivially holds.
      \item \textbf{Inductive step.} Suppose that the polytope $\cV_k$ ($k < n$) has integer vertices. We will show that the same holds for $\cV_{k+1}$. Clearly, $\cV_{k+1} = \cV_k \ext^{h_{k+1}} \cX_{k+1}$. From the properties of the structural decomposition, we know that $\cK_{k+1}$ is either the singleton $\{1\}$, or a probability simplex $\Delta^{s_{k+1}}$ for some appropriate dimension $s_{k+1}$. We break the analysis accordingly.
          \begin{itemize}
            \item If $\cX_{k+1} = \{1\}$, the scaled extension represents filling in a linearly-dependent entry in $\vec{v} \in \cV$ by summing already-filled-in entries. So, $h_{k+1}$ takes a partially-filled-in vector from $\cV_{k}$ and sums up some of its coordinates. Let $\vec{v}_1,\dots,\vec{v}_n$ be the vertices of $\cV_k$. Using \cref{lem:vertices of scext}, the vertices of $\cV_{k+1}$ are a subset of
                \begin{equation}\label{eq:sup 1}
                    \{(\vec{v}_i, h(\vec{v}_i)\cdot 1) : i = 1,\dots, n\}.
                \end{equation}
                Since by inductive hypothesis $\vec{v}_i$ have integer coordinates, and $h$ sums up some of them, $h(\vec{v})_i$ is integer for all $i = 1,\dots,n$. So, all of the vectors in~\eqref{eq:sup 1} have integer coordinates, and in particular this must be true of the vertices of $\cV_{k+1}$.
            \item If $\cX_{k+1} = \Delta^{s_{k+1}}$, the scaled extension represents the operation of partitioning an already-filled-in entry $v[\sigma,\tau]$ of $\cV_k$ into $s_i$ non-negative real values. The affine function $h_{k+1}$ extracts the entry $v[\sigma,\tau]$ from each vector $\vec{v} \in \cV_k$. Let $\vec{v}_1, \dots,\vec{v}_n$ be the vertices of $\cV_k$. The vertices of $\Delta^{s_{k+1}}$ are the canonical basis vectors $\{\vec{e}_1, \dots, \vec{e}_{s_{k+1}}\}$. From \cref{lem:vertices of scext}, the vertices of $\cV_{k+1}$ are a subset of
                \begin{align}
                    &\{(\vec{v}_i, h(\vec{v}_i) \vec{e}_j) : i = 1,\dots, n, j = 1, \dots, s_{k+1}\} \nonumber\\
                        &\hspace{1cm}= \{(\vec{v}_i, {v}_i[\sigma,\tau] \vec{e}_j) : i = 1,\dots, n, j = 1, \dots, s_{k+1}\}.\label{eq:sup delta}
                \end{align}
                Since by inductive hypothesis the vertices $\vec{v}_i$ have integer coordinates, $v_i[\sigma,\tau]$ is an integer. Since the canonical basis vector only have entries in $\{0,1\}$, all of the vectors in~\eqref{eq:sup delta} have integer coordinates. So, in particular, this must be true of the vertices of $\cV_{k+1}$.\qedhere
          \end{itemize}
    \end{itemize}
\end{proof}

%% file: text/appendix_experiments.tex
\section{Additional Experimental Results}\label{app:experiments}

In this section we present additional computational results. 
Specifically, we present results on how well algorithms can solve for EFCE (and thus also EFCCE and NFCCE since they are supsets of EFCE) after our new scaled-extension-based structural decomposition has been computed for the polytope of correlation plans using the algorithm that we presented in the body. 
The speed of that algorithm for computing the decomposition is extremely fast, as shown in the body both theoretically and experimentally. Here we report the performance of two leading algorithms for finding an approximate optimal EFCE after the decomposition algorithm has completed.
Specifically, we compare the performance of the regret-minimization method of~\citet{Farina19:Efficient} to that of the barrier algorithm for linear programming implemented by the Gurobi commercial linear programming solver, as described in the body of the paper. 
(On these problems, any linear programming solver could be used in principle, but simplex and dual simplex methods---even the ones in Gurobi---are prohibitively slow. Similarly, the subgradient descent method of~\citet{Farina19:Correlation} is known to be dominated by the regret-minimization method of~\citet{Farina19:Efficient}.)

Both algorithms are used to converge to a feasible EFCE---that is, no objective function was set---in the largest Goofspiel instance ($k=5$). Our implementation of the regret minimization method is single-threaded, while we allow Gurobi to use 30 threads. All experiments were conducted on a machine with 64 cores and 500GB of memory. Gurobi required roughly 200GB of memory, while the memory footprint of the regret-minimization algorithm was less than 2GB.

At all times, the regret-minimization algorithm produces \emph{feasible} correlation plans, that is, points that belong to $\Xi = \cV$. So, that algorithm's iterates' infeasibility is defined as how incentive-incompatible the computed correlation plan is, measured as the difference in value that each player would gain by optimally deviating from any recommendation at any information set in the game. In contrast, the barrier method does not guarantee that the correlation plan is primal feasible, that is, the correlation plans produced by the barrier algorithm might not be in $\Xi = \cV$. Therefore, for Gurobi, we measure infeasibility as the maximum between (i) the (maximum) violation of the constraints that define $\cV$, and (ii) the incentive-incompatibility of the iterate.

\cref{fig:goof5} shows the results. The regret minimization algorithm works better as an anytime algorithm and leads to lower infeasibility for most of the run. The barrier method needs significant time to preprocess before even the first iterates are found. After that it converges rapidly.

\begin{figure}[H]
  \centering
  \includegraphics[scale=.85]{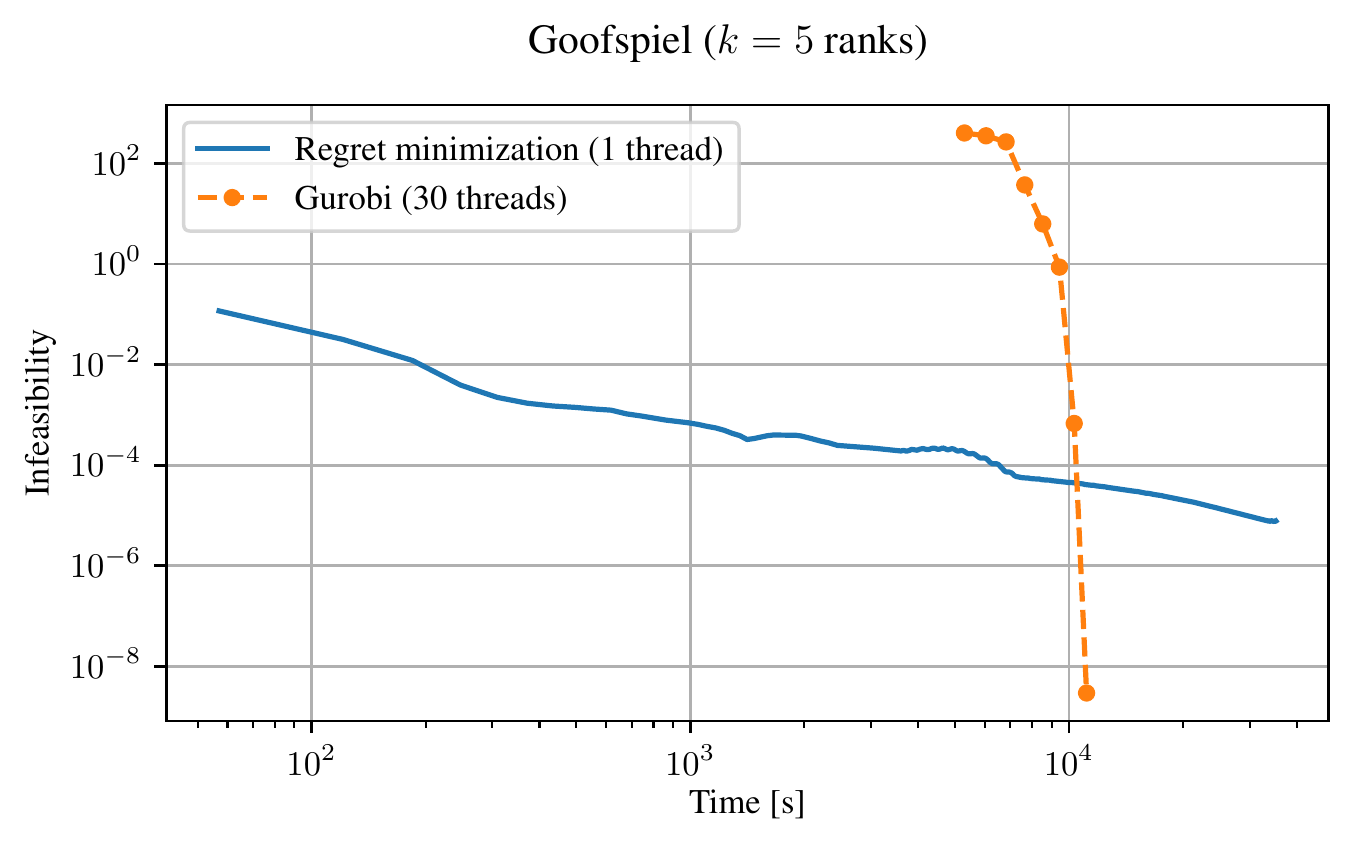}
  \caption{Performance of the regret minimization method of~\citet{Farina19:Efficient} compared to Gurobi's barrier method in the largest Goofspiel game ($k=5$).}\label{fig:goof5}
\end{figure} 